\documentclass[11pt, onecolumn]{IEEEtran}

\usepackage{fullpage}
\usepackage{microtype}
\usepackage[ruled]{algorithm}
\usepackage{algpseudocode}
\usepackage{ioa_code}
\usepackage{amssymb,amsmath,multicol}
\usepackage{amsthm}
\usepackage{algorithm}
\usepackage{algpseudocode}
\usepackage{lipsum}

\usepackage{amsfonts}
\usepackage{epsfig}
\usepackage{subfigure}
\usepackage{calc}
\usepackage{color}
\usepackage[all]{xy}
\usepackage{bm}
\usepackage{enumerate}
\usepackage{hyperref}
\usepackage{pdfsync}
\usepackage{framed}
\usepackage{float}
\usepackage{cite}
\usepackage{lipsum}

\renewcommand{\baselinestretch}{0.99}
\usepackage{graphicx}

\newcommand{\algrule}[1][.2pt]{\par\vskip.5\baselineskip\hrule height #1\par\vskip.5\baselineskip}

\newcommand{\prf}[1]{{}}

\newtheorem{Def}{Definition}[section]

\newcommand{\sacode}[5]
{ \vspace{.06in} \hrule \vspace{.06in} 
 \noindent {\bf #1}: \\
 \footnotesize \noindent {\bf Signature:}\B \nobreak
 \normalsize \begin{quote} \nobreak #2 \end{quote}
 \footnotesize \noindent {\bf States:}\B \nobreak
 \begin{quote} \nobreak #3 \end{quote}
 \noindent {\bf Transitions:} \nobreak
 \vspace{-.2in} \nobreak
 \normalsize #4
 \vspace{-.06in} \hrule \vspace{.06in} 
}

\newcommand{\act}[1]{%
    \relax\ifmmode
        \mathord{\mathcode`\-="702D\sf #1\mathcode`\-="2200}%
    \else
        $\mathord{\mathcode`\-="702D\sf #1\mathcode`\-="2200}$%
    \fi
}

\newcommand{\tup}[1]{%
    \relax\ifmmode
      \langle #1 \rangle%
    \else
        $\langle$#1$\rangle$%
    \fi
}

\newcommand{\seq}[1]{%
    \relax\ifmmode
      \langle \! \langle #1 \rangle \! \rangle%
    \else
        $\langle \! \langle$ #1 $\rangle \! \rangle$%
    \fi
}

\newcommand{\B}{\vspace*{-\smallskipamount}}

\newcommand{\T}{\hspace*{1em}}

\newcommand{\Section}[1]{\section{#1}}

\newcommand{\ms}[1]{%
    \relax\ifmmode
        \mathord{\mathcode`\-="702D\it #1\mathcode`\-="2200}%
    \else
        {\it #1}%
    \fi
}

\newcommand{\lit}[1]{%
    \relax\ifmmode
        \mathord{\mathcode`\-="702D\sf #1\mathcode`\-="2200}%
    \else
        {\it #1}%
    \fi
}

\newcommand{\XDK}[1]{}
\newcommand{\remove}[1]{} 
\newcommand{\uselater}[1]{} 

\renewcommand{\setminus}{-}

\makeatletter
\def\mainlistofsymbols{
  \normalsize
  \vspace*{1.5 em}
  \@starttoc{los}
}

\def\partonelistofsymbols{
  \normalsize
  \vspace*{1.5 em}
  \@starttoc{p1los}
}

\def\parttwolistofsymbols{
  \normalsize
  \vspace*{1.5 em}
  \@starttoc{p2los}
}

\def\l@symbol#1#2{\addpenalty{-\@highpenalty} \vskip 4pt plus 2pt
{\@dottedtocline{0}{0em}{8em}{#1}{#2}}}
\makeatother

\newcommand{\newhiddensym}[2]{%
}

\newcommand{\algIOA}[2]{\ifmmode{\text{#1}_{#2}}\else{$\text{#1}_{#2}$}\fi}

\newcommand{\EX}{\ifmmode{\xi}\else{$\xi$}\fi}
\newcommand{\EXF}{\ifmmode{\phi}\else{$\phi$}\fi}

\newcommand{\inter}[1]{
	\ifmmode{\left(\bigcap_{\mathcal{Q}\in#1}\mathcal{Q}\right)}
	\else{$\left(\bigcap_{\mathcal{Q}\in#1}\mathcal{Q}\right)$}
	\fi
}

\newcommand{\op}{\pi}

\mathchardef\mhyphen="2D

\newcommand{\vid}[1]{\ifmmode{\nu_{#1}}\else{$\nu_{#1}$}\fi}

\newcommand{\seen}{\ifmmode{seen}\else{$seen$}\fi}

\newcommand{\maxts}[1]{\ifmmode{maxTS_{#1}}\else{$maxTS_{#1}$}\fi}
\newcommand{\maxtag}[1]{\ifmmode{maxTag_{#1}}\else{$maxTag_{#1}$}\fi}
\newcommand{\maxpair}[1]{\ifmmode{maxMPair_{#1}}\else{$maxMPair_{#1}$}\fi}
\newcommand{\mintag}[1]{\ifmmode{minTag_{#1}}\else{$minTag_{#1}$}\fi}
\newcommand{\maxps}{\ifmmode{maxPS}\else{$maxPS$}\fi}
\newcommand{\conftg}[1]{\ifmmode{confirmed_{#1}}\else{$confirmed_{#1}$}\fi}
\newcommand{\maxconftag}{\ifmmode{\ms{maxCT}}\else{$maxCT$}\fi}

\newtheorem{theorem}{Theorem}[section]
\newtheorem{lemma}[theorem]{Lemma}
\newtheorem*{nono-theorem}{Theorem Satement}
\newtheorem*{nono-lemma}{Lemma Satement}
\newtheorem{note}{Remark}
\newtheorem{definition}{Definition}

\newcommand{\PutData}{{ \it{put-data}}}
\newcommand{\PutTag}{{ \it{put-tag}}}
\newcommand{\GetData}{{ \it{get-data}}}
\newcommand{\PutDataResp}{{ \it{put-data-resp}}}
\newcommand{\GetDataResp}{{ \it{get-data-resp}}}

\newcommand{\writetoLtwo}{{ \it{write-to-L2}}}
\newcommand{\readfromLtwo}{{ \it{regenerate-from-L2}}}
\newcommand{\readfromLtworesp}{{ \it{regenerate-from-L2-resp}}}
\newcommand{\readfromLtwocomplete}{{ \it{regenerate-from-L2-complete}}}
\newcommand{\GetTag}{{ \it{get-tag}}}
\newcommand{\GetTagResp}{{ \it{get-tag-resp}}}
\newcommand{\GetCommitedTag}{{ \it{get-commited-tag}}}
\newcommand{\GetCommitedTagResp}{{ \it{get-commited-tag-resp}}}

\newcommand{\PutDateResp}{{ \it{put-data-resp}}}
\newcommand{\PutTagResp}{{ \it{put-tag-resp}}}
\newcommand{\WriteAckResp}{{ \it{write-to-L2-complete}}}
\newcommand{\broadcastResp}{{ \it{broadcast-resp}}}

\newcommand{\queryData}{{\sc query-data}}
\newcommand{\queryTag}{{\sc query-tag}}
\newcommand{\queryCommitedTag}{{\sc query-comm-tag}}
\newcommand{\reqCommitTag}{{\sc commit-tag}}
\newcommand{\writeCodeElem}{{\sc write-code-elem}}
\newcommand{\queryCodeElem}{{\sc query-code-elem}}
\newcommand{\sendCodeElem}{{\sc send-helper-elem}}
\newcommand{\ackCodeElem}{{\sc ack-code-elem}}
\newcommand{\putTagLabel}{{\sc put-tag}}
\newcommand{\putDataLabel}{{\sc put-data}}

\newcommand{\calc}{{LDS}} 

\begin{document}
	
	\title{A Layered Architecture for Erasure-Coded Consistent Distributed Storage  }
	\author{
		\IEEEauthorblockN{Kishori M. Konwar, N. Prakash, Nancy Lynch, Muriel M{\'{e}}dard}\\
		\IEEEauthorblockA{
			Department of Electrical Engineering \& Computer Science\\ Massachusetts Institute of Technology 
			\\ Cambridge, MA, USA
			\\ \emph{\{kishori, lynch\}@csail.mit.edu, \{prakashn, medard\}@mit.edu} }
		\thanks{A shorter version of this work appears as a regular paper in the ACM Proceedings of Principles of Distributed Computing (PODC) 2017, DOI: 10.1145/3087801.3087832.
		The work is supported in part by AFOSR  under grants FA9550-13-1-0042, FA9550-14-1-043, FA9550-14-1-0403, and in part by NSF under awards CCF-1217506, CCF-0939370.}
	}
	
	
	\maketitle

\begin{abstract}
Motivated by emerging applications to the \emph{edge computing} paradigm, we introduce a two-layer erasure-coded fault-tolerant distributed storage system offering atomic access for read and write operations. In edge computing, clients interact with an edge-layer of servers that is geographically near; the edge-layer in turn interacts with a back-end layer of servers. The edge-layer provides low latency access and temporary storage for client operations, and uses the back-end layer for persistent storage. Our algorithm, termed Layered Data Storage (LDS) algorithm, offers several features suitable for edge-computing systems, works under asynchronous message-passing environments, supports multiple readers and writers, and can tolerate $f_1 < n_1/2$ and $f_2 < n_2/3$ crash failures in the two layers having $n_1$ and $n_2$ servers, respectively. We use a class of erasure codes known as regenerating codes for storage of data in the back-end layer. The choice of regenerating codes, instead of popular choices like Reed-Solomon codes, not only  optimizes the cost of back-end storage, but also helps in optimizing communication cost of read operations, when the value needs to be recreated all the way from the back-end. The two-layer architecture permits a modular  implementation of atomicity and erasure-code protocols; the implementation of erasure-codes is mostly limited to interaction between the two layers. We prove liveness and atomicity of LDS, and also compute performance costs associated with read and write operations. In a system with $n_1 = \Theta(n_2), f_1 = \Theta(n_1), f_2 = \Theta(n_2)$, the write and read costs are respectively given by $\Theta(n_1)$ and $\Theta(1)  + n_1\mathcal{I}(\delta > 0)$. Here $\delta$ is a parameter closely related to the number of write operations that are concurrent with the read operation, and $\mathcal{I}(\delta > 0)$ is $1$ if $\delta > 0$, and $0$ if $\delta = 0$. The cost of persistent storage in the back-end layer is $\Theta(1)$. The impact of temporary storage is minimally felt in a multi-object system running $N$ independent instances of \calc, where only a small fraction of the objects undergo concurrent accesses at any point during the execution. For the multi-object system, we identify a condition on the rate of concurrent writes in the system such that the overall storage cost is dominated by that of persistent storage in the back-end layer, and is given by $\Theta(N)$.
\end{abstract}

\section{Introduction}
We introduce a two-layer erasure-coded fault-tolerant distributed storage system offering atomic access \cite{Lynch1996} for read and write operations. Providing consistent access to stored data is a fundamental problem in distributed computing. The most desirable form of consistency is atomicity, which in  simple terms, gives the users of the data service  the impression that the various concurrent read and write operations take place sequentially. Our work is motivated by applications to decentralized \emph{edge computing}, which is an emerging distributed computing paradigm where processing of data moves closer to the users instead of processing the entire data in distant data centers or cloud centers~\cite{Bonomi2012,EvansCisco,AkamaiNygren:2010,EdgeComputing:2016}. Edge computing is considered to be a key enabler for Internet of Things. In this form of computing, the users  or clients interact with servers in the edge of the network, which forms the first layer of servers. The edge servers in turn interact with a second layer servers in the back-end, which is either a distant data-center or a cloud center. Geographic proximity of edge servers to clients permits high speed operations between clients and the edge layer, whereas communication between the edge and the back-end layer is typically much slower~\cite{AkamaiNygren:2010}. Thus, it is desirable whenever possible to complete client operations via interaction only with the edge layer. The edge servers however are severely restricted in their total storage capacity. We envisage a system that handles millions of files, which we call objects; the edge servers clearly do not  have the capacity to store all the objects for the entire duration of execution. In practice, at any given time, only a tiny fraction of all objects undergo concurrent accesses; in our system, the limited storage space in the edge layer acts as a temporary storage for those objects that are getting accessed. The second layer of servers provide permanent storage for all the objects for the entire duration of execution. The servers in the first layer act as virtual clients of the second layer servers.

An important  requirement in edge-computing systems is to reduce the cost of operation of the back-end layer, by making efficient use of the edge layer~\cite{EdgeComputing:2016}. Communication between the two layers, and persistent storage in the second layer contribute to the cost of operation of the second layer. We address both these factors in our system design. The layered approach to implementing an atomic storage service carries the advantage that, during intervals of high concurrency from write operations on any one object, the edge layer can be used to \emph{retain} the more recent versions of the object that are being (concurrently) written, while \emph{filtering} out the outdated versions. The ability to avoid writing every version to the second layer decreases the overall write communication cost between the two layers. Our architecture also permits the edge layer to be configured as a \emph{proxy cache} layer for objects that are frequently read, and thus avoids the need to read from the back-end layer for such objects. 

In this work, we use a recent class of erasure codes known as \emph{regenerating codes}~\cite{dimakis} for storage of data in the back-end layer. From a storage cost view-point, these are as efficient as popular erasure codes like Reed-Solomon codes~\cite{rscodes}. In our system, usage of regenerating codes, instead of Reed-Solomon codes, provides the extra advantage of reducing read communication cost when the object needs to be recreated from the coded data in the cloud layer. Specifically, we rely on class of regenerating codes known as \emph{minimum bandwidth regenerating codes} for simultaneously optimizing read and storage costs.

While this may be the first work that explicitly uses regenerating codes for consistent data storage, the study of erasure codes---like Reed-Solomon codes---in implementations of consistent distributed storage, is an active area of research by itself~\cite{DGL08, CadambeLMM14, SODA2016, kedar_bounds, cadambe_podc}. In the commonly used single-layer storage systems, for several regimes of operation, cost metrics of Reed-Solomon-code based implementations~\cite{aguilera, DGL08, CadambeLMM14, SODA2016} outperform those of replication based implementations~\cite{ABD96}. In comparison with single layer systems, the layered architecture naturally permits a layering of the protocols needed to implement atomicity, and erasure code in the cloud layer. The protocols needed to implement atomicity are largely limited to interactions between the clients and the edge servers, while those needed to implement the erasure code are largely limited to interactions between the edge and cloud servers. From an engineering viewpoint, the modularity of our implementation  makes it suitable even for situations that does not necessarily demand a two-layer system.

\subsection{Our Algorithm for the Two-Layer System} 
We propose the \emph{Layered Distributed Storage} (LDS) algorithm for implementing a multi-writer, multi-reader atomic storage service over a two-layer asynchronous network. The algorithm is designed to address the various requirements described above for edge computing systems. A write operation completes after writing the object value to the first layer; it does not wait for the first layer to store the corresponding coded data in the second layer. For a read operation, concurrency with write operations increases the chance of it being served directly from the first layer; otherwise, servers in the first layer \emph{regenerate} coded data from the second layer, which are then relayed to the reader. Servers  in the first layer interact with those of second layer via the well defined actions (which we call as \emph{internal operations}) \writetoLtwo~ and \readfromLtwo~ for implementing the regenerating code in the second layer.  The algorithm is designed to tolerate $f_1 < n_1/2$ and $f_2 < n_2/3$ crash failures in the first and second layers, having $n_1$ and $n_2$ servers, respectively. We prove liveness and atomicity properties of the algorithm, and also calculate various performance costs. In a system with $n_1 = \Theta(n_2), f_1 = \Theta(n_1), f_2 = \Theta(n_2)$, the write and read costs are respectively given by $\Theta(n_1)$ and $\Theta(1)  + n_1\mathcal{I}(\delta > 0)$. Here $\delta$ is a parameter closely related to the number of write or internal \writetoLtwo~ operations that are concurrent with the read operation, and $\mathcal{I}(\delta > 0)$ is $1$ if $\delta > 0$, and $0$ if $\delta = 0$.. Our ability to reduce the read cost to $\Theta(1)$, when $\delta = 0$ comes from the usage of {minimum bandwidth regenerating} (MBR) codes (see Section \ref{sec:models}). In order to ascertain the contribution of temporary storage cost to the overall storage cost, we carry out a multi-object (say $N$) analysis, where each of the $N$ objects is implemented by an independent instance  of the \calc~algorithm.  The multi-object analysis assumes bounded latency for point-to-point channels. We identify conditions on the  total number of concurrent write operations per unit time, such that the permanent storage cost in the second layer dominates the temporary storage cost in the first layer, and is given by  $\Theta(N)$. 
Further, we compute bounds on completion times of successful client operations, under bounded latency.

\subsection{Related Work} 
Replication based algorithms for implementing atomic shared memory appears in \cite{ABD96}, \cite{FL03}. The model in \cite{FL03} uses a two-layer system, with one layer dedicated exclusively for meta-data, and other layer for storage. The model is suitable when actual data is much larger than meta-data, and permits easy scalability of the storage layer. However, clients interact with servers in both layers, and thus is not directly comparable to our model, where clients only interact with the first layer. Both \cite{ABD96}, \cite{FL03} use quorums for implementing atomicity; variations of these algorithms appear in practical systems like Cassandra~\cite{Lakshman:2010}.  Replication based algorithms in single-layer systems, for dynamic settings appear in RAMBO~\cite{LS02}, DynaStore~\cite{AKMS11}. Dynamic setting allow servers to leave and enter the system; these algorithms rely on reconfiguration of quorums. Erasure-code based implementations of consistent data storage in single layer systems appear in ~\cite{aguilera, DGL08, CadambeLMM14, SODA2016, kedar_bounds}. Bounds on the performance costs for erasure-code based implementations appear in \cite{cadambe_podc, kedar_bounds}. In \cite{cachin}, \cite{dobre}, \cite{hendricks}, erasure codes are used in algorithms for implementing atomic memory in settings that tolerate Byzantine failures.  In \cite{DGL08, amnesic, radon}, authors provide algorithms that permit repair of crashed servers (in a static setting), while implementing consistent storage. In the content of our work, it is of future interest to develop protocols for recovery of crashed servers in the second-layer, which implements permanent coded storage.

We rely on regenerating codes which were introduced in \cite{dimakis} with the motivation of  enabling efficient repair of failed servers in distributed storage systems. For the same storage-overhead and resiliency, the communication cost for repair, termed repair-bandwidth, is substantially less than what is needed by popular codes like Reed-Solomon codes. There has been significant theoretical progress since the work of \cite{dimakis}; a survey appears in  \cite{dimakis2011survey}.  Several systems works show usefulness of these codes or their variations in practical systems for immutable data~\cite{sathiamoorthy, rashmi_fast15, hotstor}. In this work, we cast internal read operations by virtual clients in the first layer as repair operations, and this enables us to reduce the overall read cost.  We rely on code constructions from \cite{prod_matrix} for the existence of MBR codes needed in our work, and these codes offer \emph{exact repair}. A different class of codes known as Random Linear Network Codes~\cite{rlnc} permit implementation of regenerating codes via  \emph{functional repair}. These codes offer probabilistic guarantees, and permit near optimal operation of regenerating codes for any choice of operating point suggested by \cite{dimakis}. In the context of our work, it will be interesting to find out the probabilistic guarantees that can be obtained if we use RLNCs instead of the codes in \cite{prod_matrix}.

System Model and definitions appear in Section \ref{sec:models}. The pseudo code of the \calc~algorithm, along with its description is presented in Section \ref{sec:algo}. In Section \ref{sec:propeties}, we state several properties of the algorithm, which are tied together to prove its liveness and atomicity properties. Performance cost analysis appears in Section \ref{sec:costs}. Our conclusions appear in Section \ref{sec:conc}. Proofs of various claims appear in the Appendix.

\Section{System Model and Definitions} \label{sec:models}

\paragraph{Model of Computation} We assume a distributed storage system consisting of asynchronous processes of three types: writers (${\mathcal W}$), readers (${\mathcal R}$) and servers (${\mathcal S}$). The servers are organized into two logical layers $\mathcal{L}_1$ and $\mathcal{L}_2$, with $\mathcal{L}_i$ consisting of $n_i$, $i = 1, 2$ servers. Each process has a unique id, and the ids are totally ordered. Client (reader/writer) interactions are limited to servers in $\mathcal{L}_1$, the servers in $\mathcal{L}_1$ in turn interact with servers in $\mathcal{L}_2$. Further, the servers in $\mathcal{L}_1$ and $\mathcal{L}_2$ are denoted by $\{s_1, s_2, \ldots, s_{n_1}\}$ and $\{s_{n_1 + 1}, s_{n_1+ 2}, \ldots, s_{n_1 + n_2}\}$, respectively. We assume the clients to be well-formed, i.e., a client issues a new operation only after completion of its previous operation, if any. The $\mathcal{L}_1$-$\mathcal{L}_2$ interaction happens via the well defined actions \writetoLtwo~ and \readfromLtwo. We will refer to these actions as internal operations initiated by the servers in $\mathcal{L}_1$. We assume a crash failure model for processes. Once a process crashes, it does not execute any further steps for the rest of the execution. The \calc~algorithm is designed to tolerate $f_i$ crash failures in layer $\mathcal{L}_i, i = 1, 2$, where $f_1 < n_1/2$ and $f_2 < n_2/3$. Any number of readers and writers can crash during the execution. Communication is modeled via reliable point-to-point links between any two processes. This means that as long as the destination process is non-faulty, any message sent on  the link is guaranteed to eventually reach the destination process.  The model allows the sender process to fail after placing the message in the channel; message-delivery depends only on whether the destination is non-faulty.

\paragraph{Liveness and Atomicity}
We implement one object, say $x$, via the \calc~algorithm supporting read/write operations. 
For multiple objects, we simply run multiple instances of the \calc~algorithm. 
The object value $v$ come from the set $\mathcal{V}$; initially $v$ is set to a distinguished value $v_0$ ($\in \mathcal{V}$). Reader $r$ requests a read operation on  object $x$. Similarly, a write operation is requested by a writer $w$. Each operation at a non-faulty client begins with an \emph{invocation step} and terminates with a  \emph{response step}. An operation $\op$ is \emph{incomplete} in an execution when the invocation step of $\op$ does not have the associated response step; otherwise we say that $\op$ is \emph{complete}.  In an execution, we say that an operation (read or write) $\op_1$ {\em precedes} another operation $\op_2$,  if the response step for $\pi_1$ precedes the invocation step of $\pi_2$.  Two operations are {\em concurrent} if neither precedes the other. 

By liveness, we mean that during any well-formed execution of the algorithm,  any read or write operation initiated by a non-faulty reader or writer completes,  despite the crash failure of any other client and up to $f_1$ server crashes in $\mathcal{L}_1$, and up to $f_2$ server crashes in $\mathcal{L}_2$. By atomicity of an execution, we mean that the read and write operations in the execution can be arranged in a sequential order that is consistent with the order of invocations and responses. We refer to \cite{Lynch1996} for formal definition of atomicity. We use the sufficient condition presented in Lemma $13.16$ of \cite{Lynch1996} to prove atomicity of \calc.

\paragraph{Regenerating Codes}
We introduce the framework as in \cite{dimakis}, and then see its usage in our work. In the regenerating-code framework, a file $\mathcal{F}$  of size $B$ symbols is encoded and stored across $n$ servers such that each server stores $\alpha$ symbols. The symbols are assumed to be drawn from a finite field $\mathbb{F}_q$, for some $q$. The content from any $k$ servers ($k\alpha$ symbols) can be used to decode the original file $\mathcal{F}$. For repair of a failed server, the replacement server contacts any subset of $d \geq k$ surviving servers in the system, and downloads $\beta$ symbols from each of the $d$ servers. The $\beta$ symbols from a \emph{helper} server is possibly a function of the $\alpha$ symbols in the server.  The parameters of the code, say $\mathcal{C}$, shall be denoted as $\{(n, k, d)(\alpha, \beta)\}$. It was shown in \cite{dimakis} that the file-size $B$ is upper bounded by $B \leq \sum_{i=0}^{k-1}\min(\alpha, (d-i)\beta)$. Two extreme points of operation correspond to the minimum storage overhead (MSR) operating point, with $B = k\alpha$ and minimum repair bandwidth (MBR) operating point, with $\alpha = d\beta$.  In this work, we use codes at the MBR operating point. The file-size at the MBR point is give by $B_{MBR} = \sum_{i=0}^{k-1}(d-i)\beta$. We also focus on exact-repair codes, meaning that the content of a replacement server after repair is identical to what was stored in the server before crash failure (the framework permits \emph{functional repair}~\cite{dimakis} which we do not consider). Code constructions for any set of parameters at the MBR point appear in \cite{prod_matrix}, and we rely on this work for existence of codes. In this work, the file $\mathcal{F}$ corresponds to the object value $v$ that is written.

In this work, we use an $\{(n = n_1 + n_2, k, d)(\alpha, \beta)\}$ MBR code $\mathcal{C}$. The parameters $k$ and $d$ are such that $n_1 = 2f_1 + k$ and $n_2 = 2f_2 + d$. We define two additional codes $\mathcal{C}_1$ and $\mathcal{C}_2$ that are derived from the code $\mathcal{C}$. The code $\mathcal{C}_1$ is obtained by restricting attention to the first $n_1$ coded symbols of $\mathcal{C}$, while the code $\mathcal{C}_2$ is obtained by restricting attention to the last $n_2$ coded symbols of $\mathcal{C}$. Thus if $[c_1 \ c_2 \ \ldots c_{n_1} \ c_{n_1 + 1} \ \ldots c_{n_1 + n_2}], c_i \in \mathbb{F}_q^{\alpha}$ denotes a codeword of $\mathcal{C}$, the vectors $[c_1 \ c_2 \ \ldots c_{n_1} ]$ and $[c_{n_1 + 1} \ \ldots c_{n_1 + n_2}]$ will be codewords of $\mathcal{C}_1$ and $\mathcal{C}_2$, respectively. We associate the code symbol $c_i$ with server $s_i$, $1 \leq i \leq n_1 + n2$.

The usage of these three codes is as follows. Each server in $\mathcal{L}_1$, having access to the object value $v$ (at an appropriate point in the execution) encodes $v$ using  code $\mathcal{C}_2$ and sends coded data $c_{n_1 + i}$ to server $s_{n_1 + i}$ in $\mathcal{L}_2, 1 \leq i \leq n_2$. During a read operation, a server say $s_j$ in $\mathcal{L}_1$ can potentially reconstruct  the coded data $c_j$ using content from $\mathcal{L}_2$. Here we think of $c_j$ as part of the code $\mathcal{C}$, and $c_j$ gets reconstructed via a repair procedure (invoked by server $s_j$ in $\mathcal{L}_1$) where the $d$ helper servers belong to $\mathcal{L}_2$. By operating at the MBR point, we minimize the cost that need by the server $s_j$ to reconstruct $c_j$. Finally, in our algorithm, we permit the possibility that the reader receives $k$ coded data elements from $k$ servers in $\mathcal{L}_1$, during a read operation. In this case, the reader uses the code $\mathcal{C}_1$ to attempt decoding the object value $v$. 

An important property of the MBR code construction in \cite{prod_matrix}, which is needed in our algorithm, is the fact the a helper server only needs to know the index of the failed server, while computing the helper data, and does not need to know the indices of the other $d-1$ helpers whose helper data will be used in repair. Not all regenerating code constructions, including those of MBR codes, have this property that we need. In our work, a server $s_j \in \mathcal{L}_1$ requests for help from all servers in $\mathcal{L}_2$, and does not know a priori, the subset of $d$ servers that will form the helper servers. As we shall see in the algorithm, the server $s_j$ simply relies on the first $d$ responses that it receives, and considers these as the helper data for repair.  In this case, it is crucial that any  server in $\mathcal{L}_2$ that computes its $\beta$ symbols does so without any assumption on the specific set of $d$ servers in $\mathcal{L}_2$ that will eventually form the helper servers for the repair operation.

\paragraph{Storage and Communication Costs}

The communication cost associated with a read or write operation is the (worst-case) size of the total data that gets transmitted in the messages sent as part of the operation. While calculating write-cost, we also include costs due to internal \writetoLtwo~operations initiated as a result of the write, even though these internal \writetoLtwo~operations do not influence the termination point of the write operation. The storage cost at any point in the execution is the worst-case total amount of data that is stored in the servers in $\mathcal{L}_1$ and $\mathcal{L}_2$. The total data in $\mathcal{L}_1$ contributes to temporary storage cost, while that in $\mathcal{L}_2$ contributes to permanent storage cost. Costs contributed by meta-data (data for book keeping such as tags, counters, etc.) are ignored while ascertaining either storage or communication costs.  Further the costs are normalized by the size of $v$; in other words, costs are expressed as though size of $v$ is $1$ unit.

\section{$LDS$~Algorithm }\label{sec:algo}

\begin{algorithm*}
	\begin{algorithmic}[2]
		\begin{multicols}{2}
			{\footnotesize
				\Part{ $LDS$ steps  at  a writer $w$} {}\EndPart
				\Part{\underline{\GetTag}} {
					\State  Send  {\queryTag} to servers in ${\mathcal L}_1$
					\State  Wait for responses from  $f_1+k$ servers, and select max tag $t$. 
				}\EndPart
				
				\Part{\underline{\PutData}} {
					\State Create new tag $t_w = (t^.z + 1, w)$.  
					\State Send  ({\putDataLabel}, $(t_w, v)$) to servers in ${\mathcal L}_1$
					\State  Wait for responses from $f_1+k$  servers in ${\mathcal L}_1$, and terminate
				}	\EndPart
			}
		\end{multicols}
	\end{algorithmic}	
	\algrule  
	\begin{algorithmic}[2]
		\begin{multicols}{2}
			{\footnotesize
				\Part{ $LDS$ steps  at  a reader $r$} {}\EndPart
				\Part{\underline{\GetCommitedTag}} {
					\State  Send  {\queryCommitedTag} to servers in ${\mathcal L}_1$
					\State  Await  $f_1+k$ responses, and  select  max tag $t_{req}$
				}\EndPart
				
				\Part{\underline{\GetData}} {
					\State  Send   $(${\queryData}$, t_{req})$ to servers in ${\mathcal L}_1$
					\State  Await responses from $f_1+k$ servers such that at least one of them is a (tag, value) pair, or at least $k$ of them are (tag, coded-element) pairs corresponding to some common tag.  In the latter case, decode corresponding value using code $\mathcal{C}_1$. Select the $(t_r, v_r)$ pair corresponding to the highest tag, from the available (tag, value) pairs.  
				}\EndPart
				\Part{\underline{\PutTag}} {
					\State  Send   ({\putTagLabel}, $t_{r}$) to servers in ${\mathcal L}_1$
					\State  Await responses from $f_1+k$ servers in ${\mathcal L}_1$. Return $v_r$
				}\EndPart
			}
		\end{multicols}
	\end{algorithmic}
	
	\caption{The $LDS$ algorithm for a writer $w \in \mathcal{W}$ and  reader $r \in {\mathcal R}$.}
	\label{fig:lcas-clients}
\end{algorithm*}

			\begin{algorithm*}[!ht]
				\begin{algorithmic}[2]
					\begin{multicols}{2}
						{\footnotesize
							\Part{ $LDS$ state variables \& steps  at an ${\mathbf {\mathcal L}_1}$ server, $s_j$} {}\EndPart
							\Part{State Variables}{ 
								\Statex $L \subseteq   \mathcal{T} \times {\mathcal V}$, initially $\{ (t_0, \bot) \}$				
								\Statex $\Gamma \subseteq   \mathcal{R} \times {\mathcal T}$, initially empty		
								\Statex $t_c \in \mathcal{T}$ initially $t_c = t_0$
								\Statex $commitCounter[t]$ : $ t \in  {\mathcal T}$, initially  $0 \ \forall t \in  {\mathcal T}$
								\Statex $readCounter[r]$:   $r \in  {\mathcal R}$, initially  $0 \ \forall r \in  {\mathcal R}$ 
								\Statex $writeCounter[t]$:  $t \in  {\mathcal T}$, initially  $0 \ \forall t \in  {\mathcal T}$
								\Statex $K$ : key-value set; keys from $\mathcal{R}$, values from $\mathcal{T} \times \mathcal{H}$  	
								
							}\EndPart
							\Statex
							\Part {\underline{\GetTagResp~({\queryTag}) from $w \in \mathcal{W}$}} {
								\State send  $\max\{t:  (t, *) \in L \}$ to $w$
							}\EndPart
							
							\Statex
							\Part{ \underline{\PutDateResp~  ({\putDataLabel}, $(t_{in}, v_{in})$) received}}{
								\State  $broadcast$({\reqCommitTag}, $t_{in}$) to  ${\mathcal L}_1$
								\If{ $t_{in} > t_{c}$ } 
								\State $L \leftarrow   L \cup \{ (t_{in}, v_{in})\}$
								\Else
								\State send {\sc ACK} to writer $w$ of tag $t_{in}$
								\EndIf
							}\EndPart
							\Statex	
							\Part{ \underline{\broadcastResp~({\reqCommitTag}, $t_{in}$) received}}{
								\State  $commitCounter[t_{in}] \leftarrow commitCounter[t_{in}] +1$
								\If{  $(t_{in}, *) \in L \wedge $ $commitCounter[t_{in}]  \geq f_1 + k$} 
								\State send {\sc ACK} to writer $w$ of tag $t_{in}$
								\If{ $t_{in} > t_{c}$ } 
								\State  $t_c \leftarrow t_{in}$
								
								\State  {\bf for each} $\gamma \in \Gamma$ such that $t_c \geq \gamma.t_{req}$, 
								\Statex \T \T\T\T send $(t_{in}, v_{in})$ to reader $\gamma.r$
								\Statex \T\T\T\T $\Gamma \leftarrow \Gamma \setminus \{\gamma\}$
								\State  {\bf for each} $(t, *) \in L$ s.t. $t <t_c$ //delete older value
								\Statex \T\T\T\T  $L \leftarrow L \setminus \{  (t, *)\} \cup \{  (t, \perp)\}$  
								\State   {\bf initiate} \writetoLtwo$(t_{in}, v_{in})$ // write $v_{in}$ to ${\mathcal   L}_2$
								\EndIf
								\EndIf
							}\EndPart
							\Statex
							\Part {\underline{\writetoLtwo$(t_{in}, v_{in})$}} {
								\State {\bf for each} $s_{n_1+i} \in {\mathcal L}_2$
								\State \T\T Compute coded element $c_{n_1+i}$ for value $v$ 
								\State \T\T send ({\writeCodeElem}, $(t_{in}, c_{n_1+i}$) to $s$
							}\EndPart
							\Statex
							\Part{ \underline{\WriteAckResp \ ({\ackCodeElem}, $t$)~received}}{
								\State $writeCounter[t] \leftarrow writeCounter[t] + 1$
								\If{$writeCounter[t]  = n_2 -f_2  $} 
								\State  $L \leftarrow L \setminus \{  (t, *)  \} \cup \{(t, \bot) \}$
								\EndIf
							}\EndPart
							\Part {\underline{\GetCommitedTagResp~({\queryCommitedTag}) from $r \in \mathcal{R}$}} {
								\State send  $t_c$ to $r$
							}\EndPart
							\Part {\underline{\GetDataResp~({\queryData}, $t_{req}$) from $r \in \mathcal{R} $}} {
								\If{ $(t_{req}, v_{req}) \in L$ }
								\State send $(t_{req}, v_{req})$ to reader    $r$
								\Else 
								\If{  $t_c > t_{req} \wedge (t_c, v_c) \in L$ }
								\State send $(t_{c}, v_{c})$ to reader $r$
								\Else
								\State $\Gamma \leftarrow \Gamma \cup \{  (r, t_{req}) \}$
								\State {\bf initiate} \readfromLtwo$(r)$     
								\EndIf
								\EndIf
							}\EndPart
							\Statex
							\Part {\underline{\readfromLtwo$(r)$}} {
								\State {\bf for each} $s \in {\mathcal L}_2$
								\State \T\T send ({\queryCodeElem}, $r$) to $s$
							}\EndPart
							\Statex 
							\Part{ \underline{\readfromLtwocomplete({\sendCodeElem}, $(r, t, h_{n_1 + i, j})$) recv}}{
								\State  $readCounter[r] \leftarrow readCounter[r] +1$
								\State $K[r] \leftarrow K[r] \cup \{ (t, h_{n_1 + i, j})  \}$
								\If {$readerCounter[r] =  n_2 - f_2 = f_2 + d $}
								\State  $(\widehat{t}, \widehat{c}_j) \leftarrow$ regenerate  highest possible tag using $K[r]$ 
								\Statex \T\T //$(\bot, \bot)$ if failed to regenerate any tag
								\State {\bf clear} $K[r]$ 
								\If{ $\widehat{c}_j \neq \bot  \wedge  \widehat{t} \geq  \gamma.t_{req} $} // where $\gamma = (r, t_{req})$ 
								\State send $(\widehat{t}, \widehat{c}_j)$ to $r$
								\Else
								\State send $(\bot, \bot)$ to $r$
								\EndIf
								\EndIf
							}\EndPart
							\Statex				
							\Part{ \underline{\PutTagResp~({\putTagLabel},  $(t_{in})$) received  from $r \in \mathcal{R}$}}{
								\State $\Gamma \leftarrow \Gamma \setminus \{ \gamma' \}$ // $\gamma' = (r, t_{req})$
								\If{ $t_{in} > t_{c}$ } 
								\State  $t_c \leftarrow t_{in}$
								\If{ $(t_{c}, v_{c}) \in L$}
								\State  {\bf for each} $\gamma \in \Gamma$  s. t. $t_c \geq \gamma.t_{req}$
								\State  \T\T send $(t_{c}, v_{c})$ to reader $\gamma.r$,
								\State \T\T  $\Gamma \leftarrow   \Gamma \setminus \{\gamma\}$.
								\State   {\bf initiate} \writetoLtwo$(t_{in}, v_{in})$
								\Else
								\State $L \leftarrow L \cup \{  (t_c, \perp)\}$
								\State  $\bar{t} \leftarrow \max \{ t:  t < t_{c} \wedge (t,  v) \in L\} $ 
								\Statex \T\T\T // $\bar{t} = \perp$, if none exists
								\State  {\bf for each} $\gamma \in \Gamma$ such that $\bar{t} \geq \gamma.t_{req}$, 
								\Statex \T\T\T\T send $(\bar{t}, \bar{v})$ to reader $\gamma.r$
								\Statex \T\T\T\T  $\Gamma \leftarrow \Gamma \setminus \{\gamma\}$
								\EndIf
								\State  {\bf for each} $(t, *) \in L$ s.t. $t <t_c$
								\Statex \T\T\T\T   $L \leftarrow L \setminus \{  (t, *)\} \cup \{  (t, \perp)\}$  
								\EndIf
								\State  send {\sc ACK} to $r$ 
								
							}\EndPart
						}
					\end{multicols}
				\end{algorithmic}	
					\caption{The $LDS$ algorithm for any  server in ${\mathcal L}_1$. }\label{fig:lcas-edge-servers}
				\end{algorithm*}

				\begin{algorithm*}[!ht]
				\begin{algorithmic}[2]
					\begin{multicols}{2}
						{	\footnotesize
							\Part{ $LDS$ state variables \& steps  at an ${\mathbf {\mathcal L}_2}$ server $s_{n_1 + i}$} {}\EndPart
							\Part{State Variables}{ 		
								\Statex  $(t, c) \in {\mathcal T} \times {\mathbb F}_q^{\alpha}$, initially $(t_0, c_0)$	
							}\EndPart
							\Statex
							\Part {\underline{ $write\act{-}to\act{-}L2\act{-}resp$ ({\writeCodeElem}, $(t_{in}, c_{in})$) from  $s_{j}$}} {
								\If{ $t_{in} > t$ }
								\State $(t,c) \leftarrow (t_{in}, c_{in})$
								\EndIf
								\State send ({\ackCodeElem}, $t_{in}$) to $s_{j}$
							}\EndPart
							\Statex	
							\Part {\underline{\readfromLtworesp({\queryCodeElem}, $r$) from  $s_{j}$}} {
								\State Compute helper data $h_{n_1 + i, j} \in \mathbb{F}_q^{\beta}$ for repairing $c_{j}$
								\State send ({\sendCodeElem}, $(r, t, h_{n_1 + i, j})$) to $s_{j}$
							}\EndPart
						} \end{multicols}
					\end{algorithmic}	
					\caption{The $LDS$ algorithm for any  server in ${\mathcal L}_2$ . }\label{fig:lcas-backend-servers}
				\end{algorithm*}

	In this section,  we  present the {$LDS$~algorithm. The protocols for clients, servers in $\mathcal{L}_1$ and servers in $\mathcal{L}_2$ appear in Figs. \ref{fig:lcas-clients}, \ref{fig:lcas-edge-servers} and \ref{fig:lcas-backend-servers} respectively. Tags are used for version control of object values. A tag $t$ is defined as a pair $(z, w)$, where $z \in \mathbb{N}$ and $w \in \mathcal{W}$ denotes the ID of a writer. We use $\mathcal{T}$ to denote the set of all possible tags. For any two tags $t_1, t_2 \in \mathcal{T}$ we say  $t_2 > t_1$ if $(i)$ $t_2.z > t_1.z$ or $(ii)$ $t_2.z = t_1.z$ and $t_2.w > t_1.w$. The relation $>$ imposes a total order on the set $\mathcal{T}$.

		Each server $s$ in $\mathcal{L}_1$ maintains the following state variables: $a)$ a list $L \subseteq \mathcal{T} \times \mathcal{V}$, which forms a temporary storage for tag-value pairs received as part of write operations, $b)$ $\Gamma \subseteq \mathcal{R} \times \mathcal{T}$, which indicates the set of readers being currently served. The pair $(r, t_{req}) \in \Gamma$ indicates that the reader $r$ requested for tag $t_{req}$ during the read operation. $c)$ $t_c$: committed tag at the server, $d)$ $K$ : a key-value set used by the server as part of internal {\readfromLtwo} operations. The keys belong to $\mathcal{R}$, and values belong to $\mathcal{T} \times \mathcal{H}$. Here $\mathcal{H}$ denotes the set of all possible helper data corresponding to coded data elements $\{c_s(v), v \in \mathcal{V}\}$.  Entries of $\mathcal{H}$ belong to $\mathbb{F}_q^{\beta}$. In addition to these, the server also maintains three counter variables for various operations. The state variable for  a server in $\mathcal{L}_2$ simply consists of one (tag, coded-element) pair. For any server $s$, we use the notation $s.y$ to refer its state variable $y$. Further, we write $s.y|_T$ to denote the value of $s.y$ at point $T$ of the execution. Following the spirit of I/O automata~\cite{Lynch1996}, an execution fragment of the algorithm is simply an alternating sequence of (the collection of all) states and actions. By an action, we mean a block of code executed by any one process without waiting for further external inputs.

				In our algorithm, we use a \emph{broadcast} primitive for certain meta-data message delivery. The primitive has the property that if the message is consumed by any one server in $\mathcal{L}_1$, the same message is eventually consumed by every non-faulty server in ${\mathcal L}_1$. An implementation of the primitive, on top of reliable communication channels (as in our model), can be found in \cite{SODA2016}.  In this implementation, the idea is that the process that invokes the broadcast protocol first sends, via point-to-point channels, the message to a fixed set $S_{f_1 + 1}$ of $f_1 + 1$ servers in  $\mathcal{L}_1$. Each of these servers, upon reception of the message for first time, sends it to all the servers in $\mathcal{L}_1$, before consuming the message itself.  The primitive helps in the scenario when the process that invokes the broadcast protocol crashes before sending the message to all servers.
				
				\subsection{Write Operation} 
				
				The write operation has two phases, and aims to temporarily store the object value $v$ in $\mathcal{L}_1$ such that up to $f_1$ failures in $\mathcal{L}_1$ does not result in loss of the value. During the first phase \GetTag, the writer $w$ determines the new tag for the value to be written. To this end, the writer queries all servers in $\mathcal{L}_1$ for maximum tags, and awaits responses from $f_1 + k$  servers in $\mathcal{L}_1$. Each server that gets the \GetTag~ request responds with the maximum tag  present in the list $L$, i.e. $\max_{(t, *) \in L}t$. The writer picks the maximum tag, say $t$, from among the responses, and creates a new and higher tag $t_w = tag(\pi)$.
				
				In the second phase \PutData, the writer sends the new (tag, value) pair to all severs in $\mathcal{L}_1$, and awaits acknowledgments from $f_1 + k$  servers in $\mathcal{L}_1$. A server that receives the new (tag, value) pair $(t_{in}, v_{in})$ as a first step uses the \emph{broadcast} primitive to send a data-reception message to all servers in $\mathcal{L}_1$. Note that this message contains only meta-data information and not the actual value. Each server $s \in \mathcal{L}_1$ maintains a committed tag variable $t_c$ to indicate the highest tag that the server either finished writing or is currently writing to $\mathcal{L}_2$. After sending the broadcast message, the server adds the pair $(t_{in}, v_{in})$ to its temporary storage list $L$ if $t_{in} > t_c$. If $t_{in} < t_c$, the server simply sends an acknowledgment back to the writer, and completes its participation in the ongoing write operation.
				If the server adds the pair $(t_{in}, v_{in})$ to $L$, it waits to hear the \emph{broadcast} message regarding the ongoing write operation from at least $f_1 + k$  servers before sending acknowledgment to the writer $w$. We implement this via the \emph{broadcast-resp} phase.  It may be noted that a server sends acknowledgment to writer via the \emph{broadcast-resp} phase only if it had added the (tag, value) pair to $L$ during the \PutDataResp~phase.
				
				\subsection{Additional Steps in \broadcastResp~Phase}
				The server performs a few additional steps during the \broadcastResp~phase, and these aid ongoing read operations, garbage collection of the temporary storage, and also help to offload the coded elements to $\mathcal{L}_2$. The execution of these additional steps do not affect the termination point of the write operation. We explain these steps next.
				
				\vspace{0.05in}
				\emph{\underline{Update committed tag $t_c$}:}
				The server checks if  $t_{in}$ is greater than the committed tag $t_c$, and if so updates the $t_c$ to $t_{in}$. We note that even though the server added $(t_{in}, v_{in})$ to $L$ only after checking $t_{in} > t_c$, the committed tag might have advanced due to concurrent write operations corresponding to higher tags and thus it is possible that  $t_{in} < t_c$ when the server does the check. Also, if $t_{in} > t_c$, it cab be shown that $(t_{in}, v_{in}) \in L$. In other words the value $v_{in}$ has not been garbage collected yet from the temporary storage $L$. We explain the mechanism of garbage collection shortly.
				
				\vspace{0.05in}
				\emph{\underline{Serve outstanding read requests}:} The server sends the pair $(t_{in}, v_{in})$ to any outstanding read request whose requested tag $t_{req} \leq t_{in}$. In this case, the server also considers the read operation as being served and will not send any further message to the corresponding reader. We note this is only one of the various possibilities to serve a reader. Read operation is discussed in detail later.
				
				\vspace{0.05in}
				\emph{\underline{Garbage collection of older tags}:}
				Garbage collection happens in two ways in our algorithm. We explain one of those here, the other will be explained as part of the description of the internal \writetoLtwo~ operation. The server replaces any (tag, value) pair $(t, v)$ in the list $L$ corresponding to $t < t_c = t_{in}$ with $(t, \perp)$, and thus removing the value associated with tags which are less than the committed tag. The combination of our method of updating the committed tag and garbage collection described here ensures that during intervals of concurrency from multiple write operations, we only offload the more recent (tag, value) pairs (after the tags get committed) to the back-end layer. Also the garbage collection described here eliminates values corresponding to older write operations which might have failed during the execution (and thus, which will not get a chance to get garbage collected via the second option which we describe below). 
				
				\vspace{0.05in}
				\emph{\underline{Internal \writetoLtwo~operation}:}
				The server computes the coded elements $\{c_{n_1 + 1}, \ldots, c_{n_1+n_2}\}$ corresponding to value $v_{in}$ and sends $(t_{in}, c_{n_1+i})$ to server $s_{n_1 + i}$ in $\mathcal{L}_2$, $1 \leq i \leq n_2$. In our algorithm, each server in $\mathcal{L}_2$ stores coded data corresponding to exactly one tag at any point during the execution. A server in $\mathcal{L}_2$ that receives (tag, coded-element) pair $(t, c)$ as part of an internal \writetoLtwo~operation replaces the local pair (tag, coded-element) pair $(t_{\ell}, c_{\ell})$ with the incoming one if $t > t_{\ell}$. The \writetoLtwo~operation initiated  by server $s \in \mathcal{L}_1$ terminates after it receives acknowledgments from $f_1 + d$ servers in $\mathcal{L}_2$. Before terminating, the server also garbage collects the pair $(t_{in}, v_{in})$ from its list.

				A pictorial illustration of the events that occur as a result of an initiation of a \PutData~phase by a writer is shown in Fig. \ref{fig:writer}.

			\setcounter{figure}{3}    
			\begin{figure*}[!ht]
				\centering
				\begin{minipage}[b]{0.48\textwidth}
					\includegraphics[width=\textwidth]{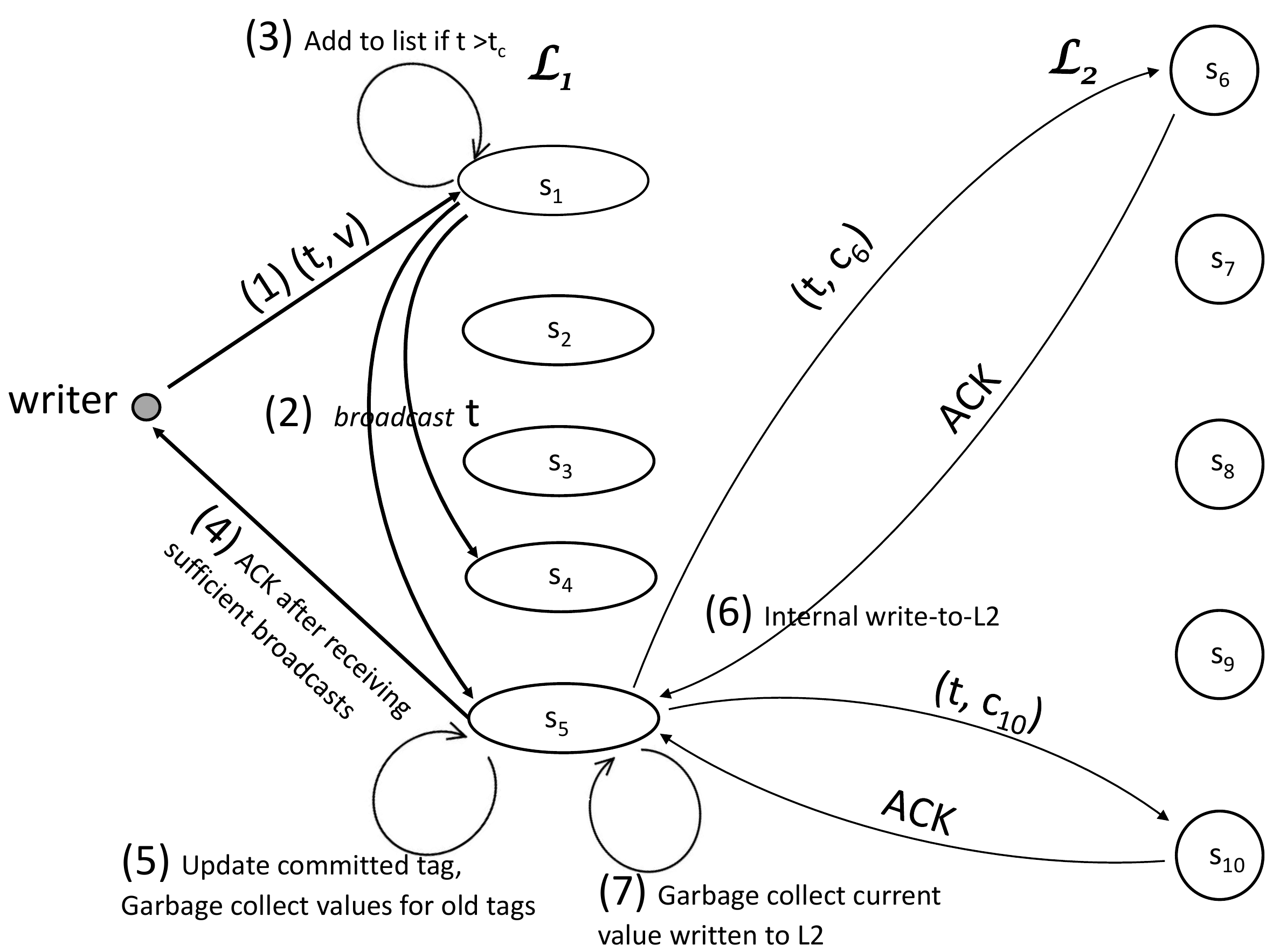}
					\caption{An illustration of the events that occur as a result of an initiation of a \PutData~phase by a writer. The illustration is only representative, and does not cover all possible executions. In this illustration, the steps occur in the order $(1), (2)$ and so on. Steps $(5), (6), (7)$ occur after sending ACK to the writer, and hence does not affect the termination point of the write operation. }
					\label{fig:writer}
				\end{minipage}
				\hfill
				\begin{minipage}[b]{0.48\textwidth}
					\includegraphics[width=\textwidth]{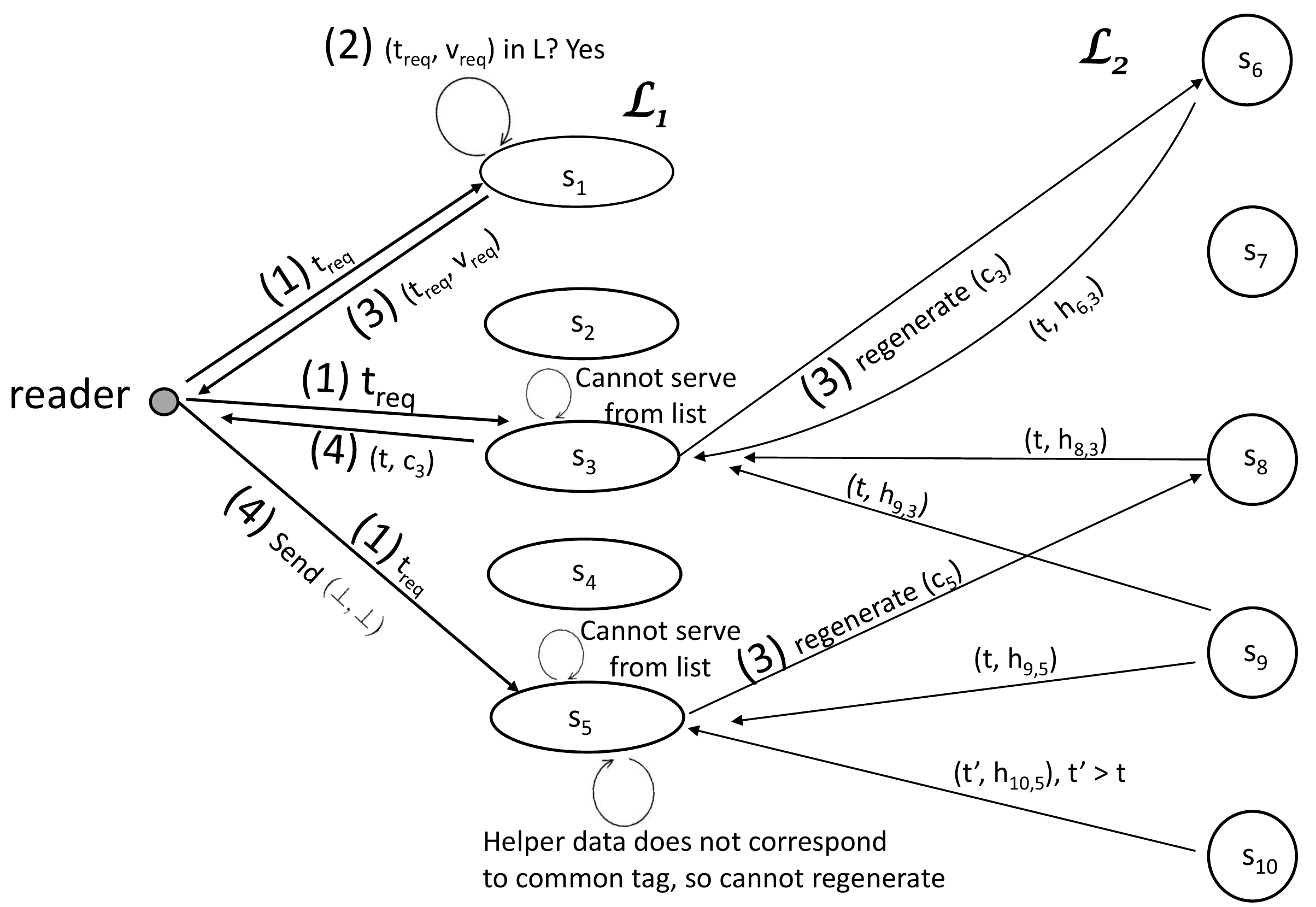}
					\caption{An illustration of the events that occur as a result of an initiation of a \GetData~phase by a reader. Once again, the illustration is only representative, and does not cover all possible executions. Server $s_1$ responds to the reader using content from local list, server $s_3$ regenerates successfully, server $s_5$ fails to regenerate successfully.                                                                                                                                             
						 } \label{fig:reader}
				\end{minipage} 
				\algrule
			\end{figure*}

				\subsection{Read operation}
				The idea behind the read operation is that the reader gets served (tag, value) pairs from temporary storage in $\mathcal{L}_1$, if it overlaps with concurrent write or internal \writetoLtwo~operations. If not, servers in $\mathcal{L}_1$ regenerate (tag, coded-element) pairs via \readfromLtwo~operations, which are then sent to the reader. In the latter case, the reader needs to decode the value $v$ using the code $\mathcal{C}_1$. A read operation consists of three phases.  During the first phase \GetCommitedTag, the reader identifies the minimum tag, $t_{req}$, whose corresponding value it can return at the end of the operation. Towards this, the reader collects committed tags from $f_1 + k$ servers in $\mathcal{L}_1$, and computes the requested tag $t_{req}$ as the maximum of these $f_1 + k$ committed tags.
				
				During the second \GetData~phase, the reader sends $t_{req}$ to all the servers in $\mathcal{L}_1$, awaits responses from $f_1 + k$ distinct servers such that $1)$ at least one of the responses contains a tag-value pair, say $(t_r, v_r), t_r \geq t_{req}$ or $2)$ at least $k$ of the responses contain coded elements corresponding to some fixed tag, say  $t_r$ such that $t_r \geq t_{req}$. In the latter case, the reader uses the code $\mathcal{C}_2$ to decode the value $v_r$ corresponding to tag $t_r$. If more than one candidate is found for the (tag, value) pair that can be returned, the reader picks the pair corresponding to the maximum tag. From the servers' point of view, a server $s \in \mathcal{L}_1$ upon reception of the \GetData~request checks if either $(t_{req}, v_{req})$ or $(t_c, v_c), t_c > t_{req}$ is in its list; in this case, $s$ responds immediately to the reader with the corresponding pair. Otherwise, $s$ adds the reader to its list $\Gamma$ of outstanding readers, initiates an internal \readfromLtwo~operation in which $s$ attempts to regenerate a tag-coded data element pair $(t', c'_{s}), t' \geq t_{req}$ via a repair process taking help from servers in $\mathcal{L}_2$. Towards this, the server $s$ contacts all servers in $\mathcal{L}_2$, and each server $\bar{s} \in \mathcal{L}_2$, using its state variable $(\bar{t}, \bar{c}_s)$, computes and sends the helper data $(\bar{t}, \bar{h}_s)$ back to the server $s$. We note that the MBR code that we use (from \cite{prod_matrix}) has the property that $\bar{h}_s$ can be uniquely computed given $\bar{c}_s$ and the id of the server $s$ which invokes the \readfromLtwo~operation. The server $s$ waits for $d + f_2$ responses, and if at least $d$ responses correspond to a common tag, say $t', t' \geq t_{req}$, regenerates the pair $(t', c')$, and sends  $(t', c')$ back to the reader. It is possible that regeneration fails, and this can happen in two ways: $1)$ the server $s$ regenerates a pair $(t'', c'')$, however $t'' < t_{req}$, $2)$ among the $f_2 + d$ responses received by the server $s$, there is no subset of $d$ responses corresponding to a common tag. In our algorithm, if regeneration from $\mathcal{L}_2$ fails, the server $s$ simply sends $(\perp, \perp)$ back to the reader. The reader interprets the response $(\perp, \perp)$ as a sign of unsuccessful regeneration.  We note that irrespective of if regeneration succeeds or not, the server does not remove the reader from its list of outstanding readers. In the algorithm, we allow the server $s$ to respond to a registered reader with a tag-value pair, during the \broadcastResp~action as explained earlier. It is possible that while the server awaits responses from $\mathcal{L}_2$ towards regeneration, a new tag $t$ gets committed by $s$ via the \broadcastResp~action; in this case, if $t \geq t_c$, server $s$ sends $(t, v)$ to $r$, and also unregisters $r$ from its outstanding reader list. A pictorial illustration of the events that occur as a result of an initiation of a \GetData~phase by a writer is shown in Fig. \ref{fig:reader}.
				
				In the third phase \PutTag, the reader \emph{writes-back}  tag $t_r$ corresponding to $v_r$, and ensures that at least $f_1 + k$ servers in $\mathcal{L}_1$ have their committed tags at least as high as $t_r$, before the read operation completes. However, the value $v_r$ is not written back in this third phase, and this is important to decrease the read cost. When a server $s \in  \mathcal{L}_1$ receives the \PutTag~request for tag $t_r$, it checks if $(t_r, v_t)$ is in its list. In this case, the server thinks of the \PutTag~request simply as a proxy for having encountered the event $commitCounter[t_r] = f_1 + k$ during the \broadcastResp~phase, and carries out all the steps that it would have done during the the \broadcastResp~phase (except sending an ACK to the writer). However, if the server sees the tag $t_r$ for the first time during the execution, it still updates its committed tag to $t_r$, and simply adds $(t_r,\perp)$ to its list.  Further, the server carries out a sequence of steps similar to the case when $(t_r, v_r) \in L$ (except initiating \writetoLtwo) before sending ACK to reader. The third phase also helps in unregistering the reader from the servers in $\mathcal{L}_1$.

				\section{Properties of the  Algorithm} \label{sec:propeties}
				We state several interesting properties of the \calc~algorithm. These will be found useful while proving the liveness and atomicity properties of the algorithm. We let $S_a \subset \mathcal{L}_1, |S_a| = f_1 + k$ to denote the set of $f_1 + k$ servers in $\mathcal{L}_1$ that never crash fail during the execution. 
				The following lemmas are only applicable to servers that are alive at the concerned point(s) of execution appearing in the lemmas.  
				
				For every operation $\pi$ in $\Pi$ corresponding to a non-faulty reader or writer, we associate a $(tag, value)$ pair that we denote as $(tag(\pi),$ $value(\pi))$. For a write operation $\pi$,  we define the $(tag(\pi), value(\pi))$ pair as the message $(t_w, v)$ which the writer sends in the {\PutData} phase. If $\pi$ is a read, we define the $(tag(\pi), value(\pi))$ pair as $(t_{r}, v)$ where $v$ is the value that gets returned, and $t_{r}$ is the  associated tag. We also define tags, in a similar manner for those failed write operations that at least managed to complete the first round of the write operation. This is simply the tag $t_w$ that the writer would use in {\PutData} phase, if it were alive. In our discussion, we ignore writes that failed before completion of the first round. 
				
				For any two points $T_1, T_2$ in an execution of \calc, we say $T_1 < T_2$ if $T_1$ occurs earlier than $T_2$ in the execution. The following three lemmas describe properties of committed tag $t_c$, and tags in the list.
				
				\begin{lemma} [{\bf Monotonicity of committed tag}]\label{lem:lcasb_prop1}
					Consider any two points $T_1$ and $T_2$ in an execution of $LDS$,  such that $T_1 < T_2$. Then, for any server $s \in \mathcal{L}_1$, $s.tc|_{T_1} \leq s.tc|_{T_2}$.
				\end{lemma}
				
				\begin{lemma}[{\bf Garbage collection of older tags}]\label{lem:tctag}
					For any server $s \in \mathcal{L}_1$, at any point $T$ in an execution of \calc, if $(t, v) \in s.L|_T$, we have $t \geq s.t_c|_T$.
				\end{lemma}
				
				\begin{lemma}[{\bf Persistence of tags corresponding to completed operations}]\label{lem:maxL_geq_tag}
					Consider any successful write or read operation $\phi$ in an execution of $LDS$, and let $T$ be any point in the execution after $\phi$ completes. For any set $S'$ of $f_1 + k$ servers in $\mathcal{L}_1$ that are non-faulty at $T$, there exists $s \in  S'$ such that $s.t_c|_T \geq tag(\phi)$ and $\max \{ t: (t, *) \in s.L|_T \}  \geq tag(\phi)$.
				\end{lemma}
				
				The following lemma shows that an internal \readfromLtwo~operation respects previously completed internal \writetoLtwo~operations. Our assumption that $f_2 < n_2/3$ is used in the proof of this lemma.
				
				\begin{lemma}[{\bf Consistency of Internal Reads with respect to Internal Writes}] \label{lem:lcasb_prop5}
					Let $\sigma_2$ denote a successful internal \writetoLtwo$(t, v)$ operation executed by some server in $\mathcal{L}_1$. Next, consider an internal \readfromLtwo \ operation $\pi_2$, initiated after the completion of $\sigma_2$, by a server $s \in \mathcal{L}_1$ such that a tag-coded-element pair, say $(t', c')$ was successfully regenerated by the  server $s$. Then, $t' \geq t$; i.e., the regenerated tag is at least as high as what was written before the read started. 
				\end{lemma}

				The following three lemmas are central to prove the liveness of read operations. 
				
				\begin{lemma}[{\bf If internal \readfromLtwo~operation fails}]\label{lem:lcasb_prop3}
					Consider an internal \readfromLtwo \ operation initiated at point $T$ of the execution by a server $s_1 \in \mathcal{L}_1$ such that $s_1$ failed to regenerate any tag-coded-element pair based on the responses. Then, there exists a point  $\widetilde{T} > T$ in the execution  such that the following statement is true:  There exists a subset $S_b$ of  $S_a$  such that $|S_b| = k$, and $\forall s' \in S_b$  $(\widetilde{t}, \widetilde{v}) \in s'.L|_{\widetilde{T}}$, where $\widetilde{t} = \max_{s \in \mathcal{L}_1}s.t_c|_{\widetilde{T}}$.
				\end{lemma}

				\begin{lemma}[{\bf If internal \readfromLtwo~operation regenerates a tag older than the request tag}]\label{lem:lcasb_prop4}
					Consider an internal \readfromLtwo \ operation initiated at point $T$ of the execution by a server $s_1 \in \mathcal{L}_1$ such that $s_1$ only manages to regenerate $(t, c)$ based on the responses, where $t < t_{req}$. Here $t_{req}$ is the tag sent by the associated reader during the \GetData~phase. Then, there exists a point  $\widetilde{T} > T$ in the execution such that the following statement is true:  There exists a subset $S_b$ of  $S_a$  such that $|S_b| = k$, and $\forall s' \in S_b$  $(\widetilde{t}, \widetilde{v}) \in s'.L|_{\widetilde{T}}$, where $\widetilde{t} = \max_{s \in \mathcal{L}_1}s.t_c|_{\widetilde{T}}$.
				\end{lemma}

				\begin{lemma}[{\bf If two Internal \readfromLtwo~operations regenerate differing tags}]\label{lem:lcasb_prop6}
					Consider internal \readfromLtwo \ operations initiated at points $T$ and $T'$ of the execution, respectively by servers $s$ and $s'$ in $\mathcal{L}_1$. Suppose that $s$ and $s'$ regenerate tags $t$ and $t'$ such that $t < t'$. Then, there exists a point  $\widetilde{T} > T$ in the execution such that the following statement is true:  There exists a subset $S_b$ of  $S_a$  such that $|S_b| = k$, and  $\forall s' \in S_b$  $(\widetilde{t}, \widetilde{v}) \in s'.L|_{\widetilde{T}}$, where $\widetilde{t} = \max_{s \in \mathcal{L}_1}s.t_c|_{\widetilde{T}}$.
				\end{lemma}

				\begin{theorem}[{\bf Liveness}]  \label{thm:lcasb_liveness}
					Consider any well-formed execution of the \calc~algorithm, where at most $f_1 < n_1/2$ and $f_2 < n_2/3$ servers crash fail in layers $\mathcal{L}_1$ and $\mathcal{L}_2$, respectively. Then every operation associated with a non-faulty client completes. 
				\end{theorem}
				
				\begin{theorem}[{\bf Atomicity}]  \label{thm:lcasb_atomicity}
					Every well-formed execution of the $LDS$ algorithm is atomic.
				\end{theorem}
				
\section{Cost Computation: Storage, Communication and Latency} \label{sec:costs}

In this section we discuss  storage and communication costs associated with read/write operations, and also carry out a latency analysis of the algorithm, in which estimates for durations of successful client operations are provided. We also analyze a multi-object system, under bounded latency, to ascertain the contribution of temporary storage toward the overall storage cost. We calculate costs for a system in which the number of servers in the two layers are of the same order, i.e., $n_1 = \Theta(n_2)$. We further assume that the parameters $k, d$ of the regenerating code are such that $k  = \Theta(n_2), d = \Theta(n_2)$. This assumption is consistent with usages of codes in practical systems. 

In this analysis, we assume that corresponding to any failed write operation $\pi$, there exists a successful write operation $\pi'$ such that $tag(\pi') > tag(\pi)$.  This essentially avoids pathological cases where the execution is a trail of only unsuccessful writes. Note that the restriction on the nature of execution was not imposed while proving liveness or atomicity. 

Like in Section \ref{sec:propeties}, our lemmas in this section apply only to servers that are non-faulty at the concerned point(s) of execution appearing in the lemmas. Also, we continue to ignore writes that failed before the completion of the first round.

\begin{lemma}[{\bf Temporary Nature of $\mathcal{L}_1$ Storage}]\label{lem:gctime1}
	Consider a successful write operation $\pi \in \beta$. Then, there exists a point of execution $T_e(\pi)$ in $\beta$ such that for all $T' \geq T_{e}(\pi)$ in $\beta$, we have  $s.t_c|_{T'} \geq tag(\pi)$ and $(t, v) \not\in s.L|_{T'}$,  $\forall s \in \mathcal{L}_1, t \leq tag(\pi)$. 
\end{lemma}

For a failed write operation $\pi \in \beta$, let $\pi'$ be the first successful write in $\beta$ such that $tag(\pi') > tag(\pi)$ (i,e., if we linearize all write operations, $\pi'$ is the first successful write operation that appears after $\pi$ - such a write operation exists because of our assumption in this section). Then, it is clear that for all $T' \geq T_{e}(\pi')$ in $\beta$, we have  $(t, v) \not\in s.L|_{T'}$,  $\forall s \in \mathcal{L}_1, t \leq tag(\pi)$, and thus Lemma \ref{lem:gctime1} indirectly applies to failed writes as well. Based on this observation, for any failed write $\pi  \in \beta$, we define the termination point $T_{end}(\pi)$ of $\pi$ as the point $T_e(\pi')$, where $\pi'$ is the first successful write in $\beta$ such that $tag(\pi') > tag(\pi)$. 

\begin{definition}[{\bf Extended write operation}]
	Corresponding to any write operation $\pi \in \beta$, we define a hypothetical extended write operation $\pi_e$ such that $tag(\pi_e) = tag(\pi)$, $T_{start}(\pi_e) = T_{start}(\pi)$ and $T_{end}(\pi_e) = \max(T_{end}(\pi), T_{e}(\pi))$, where $T_{e}(\pi)$ is as obtained from Lemma \ref{lem:gctime1}.
\end{definition}

The set of all extended write operations in $\beta$ shall be denoted by $\Pi_e$.

\begin{definition}[{\bf Concurrency Parameter $\delta_{\rho}$}]
	Consider any successful read operation $\rho \in \beta$, and let $\pi_{e}$ denote the last extended write operation in $\beta$ that completed before the start of $\rho$. Let $\Sigma = \{ \sigma_e \in \Pi_e | tag(\sigma) > tag(\pi_e) ~\mbox{and} ~\sigma_e~ \mbox{overlaps with} ~\rho  \}$. We define concurrency parameter $\delta_{\rho}$ as the cardinality of the set $\Sigma$.
\end{definition}

\begin{lemma}[{\bf Write, Read Cost}]\label{lem:writecost}
	The communication cost associated with any write operation in $\beta$ is given by 
	$n_1 + n_1n_2\frac{2d}{k(2d - k + 1)} = \Theta(n_1)$. The communication cost associated with any successful read operation $\rho$ in $\beta$ is given by $n_1(1 + \frac{n_2}{d})\frac{2d}{k(2d-k+1)} + n_1I(\delta_{\rho} > 0) = \Theta(1) + n_1I(\delta_{\rho} > 0)$. Here, $I(\delta_{\rho} > 0)$ is $1$ if $\delta_{\rho} > 0$, and $0$ if $\delta_{\rho} = 0$. 
\end{lemma}

\begin{note} \label{rem:comm}
	Our ability to reduce the read cost to $\Theta(1)$ in the absence of concurrency from extended writes comes from the usage of regenerating codes at MBR point. Regenerating codes at other operating points are not guaranteed to give the same read cost. For instance, in a system with equal number of servers in either layer, also with identical fault-tolerance (i.e., $n_1 = n_2, f_1 = f_2$), it can be shown that usage of codes at the MSR point will imply that read cost is $\Omega(n_1)$ even if $\delta_{\rho} = 0$.
\end{note}

\begin{lemma}[{\bf Single Object Permanent Storage Cost}]\label{lem:storcost}
	The (worst case) storage cost in $\mathcal{L}_2$ at any point in the execution of the \calc~algorithm  is given by $\frac{2dn_2}{k(2d-k+1)} = \Theta(1)$.
\end{lemma}

\begin{note} \label{rem:stor}
	Usage of MSR codes, instead of MBR codes, would give a storage cost of $\frac{n_2}{k} = \Theta(1)$. For fixed $n_2, k, d$, the storage-cost due to MBR codes is at most twice that of MSR codes.  As long as we focus on order-results, MBR codes do well in terms of both storage and read costs; see Remark \ref{rem:comm} as well. 
\end{note}

\subsection{Bounded Latency Analysis}

For bounded latency analysis, we assume the delay on the various point-to-point links are upper bounded as follows: $1)$ $\tau_1$, for any link between a client and a server in $\mathcal{L}_1$, $2)$ $\tau_2$, for any link between a server in $\mathcal{L}_1$ and a server in $\mathcal{L}_2$, and $3)$ $\tau_0$, for any link between two servers in $\mathcal{L}_1$. We also assume that the local computations on any process take negligible time when compared to delay on any of the links. In edge computing systems, $\tau_2$ is typically much higher than both $\tau_1$ and $\tau_0$. 

\begin{lemma}[{\bf Write, Read Latency}]\label{lem:latency}
	A successful write operation in $\beta$ completes within a duration of $4\tau_1 + 2\tau_0$. The associated extended write operation completes within a duration of $\max(3\tau_1 + 2\tau_0 + 2\tau_2, 4\tau_1 + 2\tau_0)$. A successful read operation in $\beta$ completes within a duration of $\max(6\tau_1 + 2\tau_2, 5\tau_1 + 2\tau_0 + \tau_2)$.
\end{lemma}

\subsubsection{Impact of Number of Concurrent Write Operations on Temporary Storage, via Multi-Object Analysis}
Consider implementing $N$ atomic objects in our two-layer storage system, via $N$ independent instances of the \calc~algorithm. The value of each of the objects is assumed to have size $1$. Let $\theta$ denote an upper bounded on the total number of concurrent extended write operations experienced by the system within any duration of $\tau_1$ time units. We show that under appropriate conditions on $\theta$, the total storage cost is dominated by that of permanent storage in $\mathcal{L}_2$. We make the following simplifying assumptions: $1)$ System is symmetrical so that $n_1 = n_2, f_1 = f_2 (\implies k = d)$ $2)$ $\tau_0 = \tau_1$, and $3)$ All the invoked write operations are successful. We note that it is possible to relax any of these assumptions and give a more involved analysis. Also, let $\mu = \tau_2/\tau_1$. 

\begin{lemma}[{\bf Relative Cost of Temporary Storage}]\label{lem:multiobject}
	At any point in the execution, the worst case storage cost in $\mathcal{L}_1$ and $\mathcal{L}_2$ are upper bounded by $\left\lceil 5 + 2\mu \right\rceil \theta n_1$ and $\frac{2Nn_2}{k+1}$. Specifically, if $\theta << \frac{Nn_2}{kn_1\mu}$, the overall storage cost is dominated by that of permanent storage in $\mathcal{L}_2$, and is given by $\Theta(N)$.
\end{lemma}

An illustration of Lemma \ref{lem:multiobject} is provided in Fig. \ref{fig:storage_cost}. In this example, we assume $n_1 = n_2 = 100, k = d = 80, \tau_2=10\tau_1$ and $\theta = 100$, and plot $\mathcal{L}_1$ and $\mathcal{L}_2$ storage costs as a function of the number $N$ of objects stored. As claimed in Lemma \ref{lem:multiobject}, for large $N$, overall storage cost is dominated by that of permanent storage in $\mathcal{L}_2$, and increases linearly with $N$. Also, for this example, we see that the $\mathcal{L}_2$ storage cost per object is less than $3$. If we had used replication in $\mathcal{L}_2$ (along with a suitable algorithm), instead of MBR codes, the $\mathcal{L}_2$ storage cost per object would have been $n_2 = 100$.

\begin{figure}
	\centering
	\includegraphics[width=70mm]{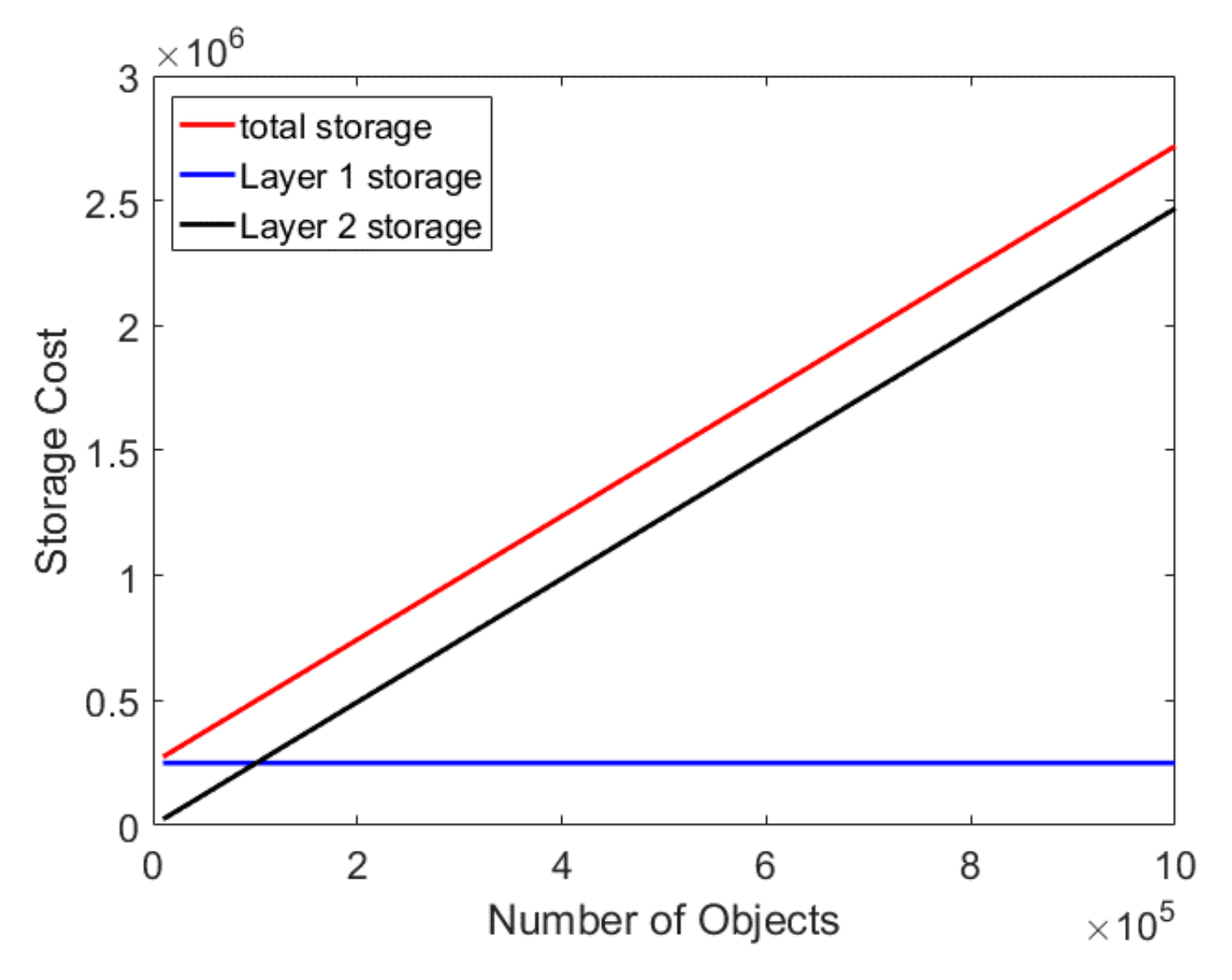}
	\caption{Illustration of the variation of $\mathcal{L}_1$ and $\mathcal{L}_2$ storage costs as a function of the number of objects stored. In this example, we assume $n_1 = n_2 = 100, k = d = 80, \tau_2=10\tau_1$ and $\theta = 100$.}
	\label{fig:storage_cost}
\end{figure}				
\section{Conclusion}\label{sec:conc}

In this paper we proposed a two-layer model for strongly consistent data-storage, while supporting read/write operations. Our model and \calc~algorithm were both motivated by the proliferation of edge computing applications. In the model, the first layer is closer (in terms of network latency) to the clients and the second layer stores bulk data. In the presence of frequent read and write operations, most of the operations are served without the need to communicate with the back-end layer, thereby decreasing the latency of operations. In this regard, the first layer behaves as a proxy cache.  In our algorithm, we use regenerating codes to simultaneously optimize storage and  read costs.  Several interesting avenues for future work exist. It is of interest to extend the framework to carry out repair of erasure-coded servers in $\mathcal{L}_2$. A model for repair in single-layer systems using erasure codes was proposed in \cite{radon}. The modularity of implementation possibly makes the repair problem in $\mathcal{L}_2$ simpler that the one in  \cite{radon}. Furthermore, we would like to explore if the modularity of implementation could be advantageously used to implement a different consistency policy like regularity without affecting the implementation of the erasure codes in the  back-end. Similarly, it is also of interest to study feasibility of other codes from the class of regenerating codes (like RLNCs~\cite{rlnc}) in the back-end layer, without affecting client protocols.

\bibliographystyle{plain}
\bibliography{biblio}

\appendices

\section*{Appendix I: Proofs of \calc~Properties}

\subsection{Proof of Lemma \ref{lem:lcasb_prop1} [{\bf Monotonicity of committed tag}]}

\begin{nono-lemma} 
	Consider any two points $T_1$ and $T_2$ in an execution of $LDS$,  such that $T_1 < T_2$. Then, for any server $s \in \mathcal{L}_1$, $s.tc|_{T_1} \leq s.tc|_{T_2}$.
\end{nono-lemma}
\begin{proof}
	For any server $s \in \mathcal{L}_1$, committed tag gets updated either via the \broadcastResp~action or the \PutTagResp~action. In both these instances, from an inspection of the algorithm, we see that the committed tag is changed to an incoming tag $t_{in}$ only if the committed tag before update is less than $t_{in}$, and this ensures that the committed tag never decreases during the execution.
\end{proof}

\subsection{Proof of Lemma \ref{lem:tctag} [{\bf Garbage collection of older tags}]}

\begin{nono-lemma}
For any server $s \in \mathcal{L}_1$, at any point $T$ in an execution of \calc, if $(t, v) \in s.L|_T$, we have $t \geq s.t_c|_T$.
\end{nono-lemma}
\begin{proof}
The statement is trivially true at the start of the execution. From an inspection of the algorithm, we see that new (tag, value) pairs are added to list only via  via the \PutDataResp~action. Further, note that a new (tag, value) pair $(t, v)$is added by server $s$ to its list via the \PutDataResp~action (say at point $T'$ of the execution) only if $t > s.t_c|_{T'}$. 

Next, note that for any server $s \in \mathcal{L}_1$, committed tag gets updated  either via the \broadcastResp~action or the \PutTagResp~action. Now, in either of these two actions, if the committed tag is updated to a new value, say $t'$, we also replace any $(t, v), t  < t'$ in the list with $(t, \perp)$. This completes the proof of the lemma.
\end{proof}

\subsection{Proof of Lemma \ref{lem:maxL_geq_tag} [{\bf Persistence of tags corresponding to completed operations}]}

\begin{nono-lemma}
Consider any successful write or read operation $\phi$ in an execution of $LDS$, and let $T$ be any point in the execution after $\phi$ completes. For any set $S'$ of $f_1 + k$ servers in $\mathcal{L}_1$ that are non-faulty at $T$, there exists $s \in  S'$ such that $s.t_c|_T \geq tag(\phi)$ and $\max \{ t: (t, *) \in s.L|_T \}  \geq tag(\phi)$.
\end{nono-lemma}
\begin{proof}
	Let us first consider the case of a successful write operation; let $w_{\phi}$ denote the writer that initiated $\phi$. Let us denote by $S_{\phi}$ the set of $f_1 + k$ servers in ${\mathcal L}_1$ whose responses were used by $w_{\phi}$ to determine termination. Since $S_{\phi}$ and $S'$ are both majorities, there exists one server\footnote{There are at least $k$ servers in the intersection, we can focus on any one of them for the proof.} $s \in S_{\phi} \cap S'$. Let $T'$ denote the point of execution where $s$ sends acknowledgment to $w_{\phi}$. We note that $s$ sends acknowledgment to $w_{\phi}$ either via the $put\act{-}data\act{-}resp$ or $broadcast\act{-}resp$ action. In either case, it is straightforward to see\footnote{By definition, the whole action occurs at one point in the execution. Recall that an action is part of code that is executed by one of the processes without waiting for any external inputs.} that state variable $s.t_c|_{T'} \geq tag(\phi)$. Clearly, $T > T'$; from Lemma \ref{lem:lcasb_prop1} it follows that $s.t_c|_{T} \geq s.t_c|_{T'} \geq tag(\phi)$. Also, from an  inspection of the algorithm, we see that any point $T''$ in the execution, the list $s.L|_{T''}$ always contains the pair $(s.t_c|_{T''}, *)$, where $*$ is either $\perp$ or the value corresponding to $s.t_c|_{T''}$. The lemma now follows (for writes) by combining the last two statements.
	
	Let us now consider the case of read operation; let $r_{\phi}$ denote the corresponding reader. Let $S_{\phi}$ denote the set of $f_1 + k$ servers in ${\mathcal L}_1$ whose responses were used by $r_{\phi}$ during the \PutTag~phase to determine termination. If $T$ denotes the point of execution when server $s \in S_{\phi}$ responded to $r_{\phi}$ via the \PutTagResp~action, we see from an inspection of the algorithm that $s.t_c|_T \geq tag(\phi)$. The rest of the proof can now be argued like in the case of writes.
\end{proof}

\subsection{Proof of Lemma \ref{lem:lcasb_prop5} [{\bf Consistency of Internal Reads with respect to Internal Writes}]}

\begin{nono-lemma}
	Let $\sigma_2$ denote a successful internal \writetoLtwo$(t, v)$ operation executed by some server in $\mathcal{L}_1$. Next, consider an internal \readfromLtwo \ operation $\pi_2$, initiated after the completion of $\sigma_2$, by a server $s \in \mathcal{L}_1$ such that a tag-coded-element pair, say $(t', c')$ was successfully regenerated by the  server $s$. Then, $t' \geq t$; i.e., the regenerated tag is at least as high as what was written before the read started. 
\end{nono-lemma}
\begin{proof}
	Let $S$ denote the set of $f_2 + d$ servers in $\mathcal{L}_2$ whose acknowledgments were	used to determine termination of $\sigma_2$, and let $S'$ denote the set of $f_2 + d$ servers in $\mathcal{L}_2$ whose responses were used by the reader to regenerate the pair $(t', c')$. Clearly, $|S \cap S'| \geq d$, since $n_2 = 2f_2 + d$. If $T$ denotes the point of execution where $\sigma_2$ completed, from an inspection of $\mathcal{L}_2$ protocols, we see that $s.t|_{T'} \geq t, \forall T' \geq T, s \in S$. 
	
	Now, recall our assumption on the system model that $f_2 < n_2/3$, and since $n_2 = 2f_2 + d$, we get that $d > f_2$. Now, for the read operation $\pi_2$ to successfully regenerate $(t', v')$, at least $d$ responses received from $S'$ must correspond to tag $t'$.
	Since $|S \backslash (S \cap S')| \leq f_2 < d$, it follows that the reader must use at least one of the responses from $S \cap S'$ while regenerating $(t', v')$. The proof now follows since $s.t|_{T'} \geq t, \forall T' \geq T, s \in S \cap S'$. 
\end{proof}

\subsection{Proof of Lemma \ref{lem:lcasb_prop3} [{\bf If internal \readfromLtwo~operation fails}]}

We need the following intermediate lemma before proving Lemma \ref{lem:lcasb_prop3}. In fact, this lemma will be also be used in the proofs of Lemmas \ref{lem:lcasb_prop4} and \ref{lem:lcasb_prop6}.

\begin{lemma}[{\bf An Intermediate Lemma}]\label{lem:intermediate}
	Consider any point $T$ in the execution, and let $\widetilde{t} = \max_{s \in \mathcal{L}_1}s.t_c|_{T}$. Let $\widetilde{T}$ denote the earliest point in the execution when the tag $\widetilde{t}$ was committed by any server in $\mathcal{L}_1$. Let $s' \in \mathcal{L}_1$ be the server which committed $\widetilde{t}$ at $\widetilde{T}$. Then, $s'$ committed the tag $\widetilde{t}$ necessarily via the  \broadcastResp({\reqCommitTag}, $\widetilde{t}$) action, and not via the \PutTagResp({\putTagLabel}, $\widetilde{t}$) action.
\end{lemma}
\begin{proof}
	We prove the lemma via contradiction. Thus, let us suppose that   $s'$ commits the tag $\widetilde{t}$ via the \PutTagResp({\putTagLabel}, $\widetilde{t}$) action. Consider the reader $r$ that initiated the corresponding \PutTag ~ phase. Let $t_{req}$ denote the tag the reader $r$ sent during the \GetData~phase. From Lemma \ref{lem:lcasb_prop1}, we know that $t_{req} < \widetilde{t}$, since $t_{req}$ is also a committed tag (the equality condition $t_{req} = \widetilde{t}$ is ruled out since $t_{req}$ is necessarily committed by a server in $\mathcal{L}_1$ at a point in the execution earlier than $\widetilde{T}$). The reader $r$, during the {\GetData} phase either received $(\widetilde{t}, \widetilde{v})$ from the one of the servers, say $\widetilde{s}$, in $\mathcal{L}_1$, or received $k$ coded elements corresponding to  tag $\widetilde{t}$ from $k$ servers in $\mathcal{L}_1$. We rule out either of these possibilities as follows. Let us first consider the possibility that $r$ received $(\widetilde{t}, \widetilde{v})$ from $\widetilde{s}$. The server $\widetilde{s}$ returns $(\widetilde{t}, \widetilde{v})$ to  $r$ via one of the following three actions: \GetDataResp, \broadcastResp ~ and \PutTagResp. Based on the arguments so far,  it is straightforward to rule out all these three possibilities via a simple inspection of the algorithm. Next, consider the case when $r$ received $k$ coded elements corresponding to  tag $\widetilde{t}$ from $k$ servers in $\mathcal{L}_1$. Each of these $k$ servers responds to $r$ via the \readfromLtwocomplete~action. In this case, it is clear that some server in $\mathcal{L}_1$ must have  committed tag $\widetilde{t}$, and initiated \writetoLtwo$(\widetilde{t}, \widetilde{v})$ before $\widetilde{T}$. Clearly, this contradicts the assumption that $\widetilde{t}$ was not committed until $\widetilde{T}$, and thus we rule out this case as well. This completes the proof of the lemma.
\end{proof}

We are now ready to prove Lemma \ref{lem:lcasb_prop3}}

\begin{nono-lemma}[{\bf If internal \readfromLtwo~operation fails}]
	Consider an internal \readfromLtwo \ operation initiated at point $T$ of the execution by a server $s_1 \in \mathcal{L}_1$ such that $s_1$ failed to regenerate any tag-coded-element pair based on the responses. Then, there exists a point  $\widetilde{T} > T$ in the execution  such that the following statement is true:  There exists a subset $S_b$ of  $S_a$  such that $|S_b| = k$, and $\forall s' \in S_b$  $(\widetilde{t}, \widetilde{v}) \in s'.L|_{\widetilde{T}}$, where $\widetilde{t} = \max_{s \in \mathcal{L}_1}s.t_c|_{\widetilde{T}}$.
\end{nono-lemma}
\begin{proof}
	Let $S_2$ denote the subset of $f_2 + d$ servers in $\mathcal{L}_2$ whose responses were used by $s_1$ to attempt regenerating a tag-coded-element pair. Also, let $T_1$ ($T_2$)   denote the earliest point in the execution when any-one (all) of the   servers in $S_2$ received the \readfromLtwo \  request from  $s_1$. Further, let $\hat{t} = \max_{s \in \mathcal{L}_1}s.t_c|_{T_2}$. Observe that no server in $\mathcal{L}_1$ completed \writetoLtwo$(\hat{t}, \hat{v})$ before $T_1$, where $\hat{v}$ is the value associated with the tag $\hat{t}$. This is because, if any server completed \writetoLtwo$(\hat{t}, \hat{v})$ before $T_1$, then $s_1$ would regenerate the coded element corresponding to $\hat{t}$, using the responses from $S_2$. This follows since, in this case, it is certain that at least $d$ servers in  $S_2$ store coded-elements corresponding to tag $\hat{t}$ during the execution fragment $[T_1 \ T_2]$. Now, let $T^*$ denote the earliest point in the execution when $\hat{t}$ was committed by any server in $\mathcal{L}_1$. The candidate for the point $\widetilde{T}$ in the lemma is $\max(T_1, T^*)$. Note that $\widetilde{T} < T_2$. Now, since $\hat{t}$ is the maximum among all committed tags at $T_2$, and since $\widetilde{T} < T_2$, it must be true  $\hat{t}$ is the also the maximum among all committed tags at $\widetilde{T}$, i.e., $\widetilde{t} = \max_{s \in \mathcal{L}_1}s.t_c|_{\widetilde{T}} = \hat{t}$. We explicitly note that no server in $\mathcal{L}_1$ completed \writetoLtwo$(\widetilde{t}, \widetilde{v})$ before $\widetilde{T}$.

	Next, we proceed to find the subset $S_b \subset S_a, |S_b| = k$, which satisfies the lemma. Toward this, let $\widetilde{s} \in \mathcal{L}_1$ be the server that committed $\hat{t}$ at $T^*$.  From  Lemma \ref{lem:intermediate}, we know that $\widetilde{s}$ committed tag  $\hat{t}$ as part of an execution of the \broadcastResp~action. In this case, we observe that at least $k$ servers (which we take as the members of $S_b$) among $S_a$ must have executed the $broadcast$({\reqCommitTag}, $\hat{t}$) step via \PutDataResp(PUT-DATA, $(\hat{t}, \hat{v})$) action; otherwise $\widetilde{s}$ would not satisfy $commitCounter[\hat{t}] \geq f_1 + k$ that is needed to commit $\hat{t}$. Also, each of these servers in $S_b$, during its corresponding   \PutDataResp(PUT-DATA, $(\hat{t}, \hat{v})$)~action, adds $(\hat{t}, \hat{v})$ to its  list. This follows since any server in  $S_b$ satisfies $\hat{t} > t_c$ at the point of the execution of the corresponding \PutDataResp(PUT-DATA, $(\hat{t}, \hat{v})$)~action. Now, it is clear that every server $s' \in S_b$ retains $(\widetilde{t}, \widetilde{v})$ in its list at $\widetilde{T}$. 
	This follows from the maximality of $\widetilde{t}$ and thus a server removes $(\widetilde{t}, \widetilde{v})$ from its list only after completion of \writetoLtwo, which has not happened at $\widetilde{T}$. This completes the proof of the lemma.
\end{proof}

\begin{note} \label{rem:Ttilde}
	In the above proof, the fact that $\widetilde{T} < T_2$ will be found useful in the proof of Lemma \ref{lem:latency}, where we compute bounds on completion times of client operations, under a bounded latency model.
\end{note}

\subsection{Proof of Lemma \ref{lem:lcasb_prop4} [{\bf If internal \readfromLtwo~operation regenerates a tag older than the request tag}]}

\begin{nono-lemma}
	Consider an internal \readfromLtwo \ operation initiated at point $T$ of the execution by a server $s_1 \in \mathcal{L}_1$ such that $s_1$ only manages to regenerate $(t, c)$ based on the responses, where $t < t_{req}$. Here $t_{req}$ is the tag sent by the associated reader during the \GetData~phase. Then, there exists a point  $\widetilde{T} > T$ in the execution such that the following statement is true:  There exists a subset $S_b$ of  $S_a$  such that $|S_b| = k$, and $\forall s' \in S_b$  $(\widetilde{t}, \widetilde{v}) \in s'.L|_{\widetilde{T}}$, where $\widetilde{t} = \max_{s \in \mathcal{L}_1}s.t_c|_{\widetilde{T}}$.
\end{nono-lemma}
\begin{proof}
	Let $S_2$ denote the subset of $f_2 + d$ servers in $\mathcal{L}_2$ whose responses were used by $s_1$ to regenerate $(t, c)$. Also, let $T_1$  denote the earliest point in the execution when any one of the servers in $S_2$ received the \readfromLtwo \  request from  $s_1$. Now, let $\hat{t} = \max_{s \in \mathcal{L}_1}s.t_c|_{T_1}$. Clearly, $\hat{t} \geq t_{req}$, since $t_{req}$ is also a committed tag. Observe that no-server in $\mathcal{L}_1$ completed \writetoLtwo$(\hat{t}, \hat{v})$ before $T_1$; otherwise $s_1$ would have regenerated a tag that is at least as high as $\hat{t}$. The reasons for the last statement are similar\footnote{ Technically, we cannot directly apply Lemma \ref{lem:lcasb_prop5} since the two internal operations potentially overlap during the execution. Lemma \ref{lem:lcasb_prop5} assume that these operations do not overlap, even though an overlap as occurring in this proof does not affect the result of Lemma \ref{lem:lcasb_prop5}.} to those used in the proof of Lemma \ref{lem:lcasb_prop5}. The candidate for   $\widetilde{T}$ is $T_1$, and thus  $\widetilde{t} = \max_{s \in \mathcal{L}_1}s.t_c|_{\widetilde{T}} = \hat{t}$.  We explicitly note that no server in $\mathcal{L}_1$ completed \writetoLtwo$(\widetilde{t}, \widetilde{v})$ before $\widetilde{T}$.
	
	The rest of the proof, where we find the subset $S_b$, which satisfies the lemma, is similar to the proof of Lemma \ref{lem:lcasb_prop3}.
\end{proof}

\subsection{Proof of Lemma \ref{lem:lcasb_prop6} [{\bf If two Internal \readfromLtwo~operations regenerate differing tags}]}

\begin{nono-lemma}
	Consider internal \readfromLtwo \ operations initiated at points $T$ and $T'$ of the execution, respectively by servers $s$ and $s'$ in $\mathcal{L}_1$. Suppose that $s$ and $s'$ regenerate tags $t$ and $t'$ such that $t < t'$. Then, there exists a point  $\widetilde{T} > T$ in the execution such that the following statement is true:  There exists a subset $S_b$ of  $S_a$  such that $|S_b| = k$, and  $\forall s' \in S_b$  $(\widetilde{t}, \widetilde{v}) \in s'.L|_{\widetilde{T}}$, where $\widetilde{t} = \max_{s \in \mathcal{L}_1}s.t_c|_{\widetilde{T}}$.
\end{nono-lemma}
\begin{proof}
	Let $S_2 \subset \mathcal{L}_2, |S_2| = f_2 + d$ denote the set of $f_2 + d$ servers in $\mathcal{L}_2$ whose responses were used by $s$ to regenerate $(t, c)$. Also, let $T_1$  denote the earliest point in the execution when any-one of the  servers in $S_2$ received the \readfromLtwo \  request from  $s$.	We observe  that no server in $\mathcal{L}_1$ completed \writetoLtwo$(t', v')$ before $T_1$; else $s$ would regenerate $(t', c')$ using the responses from $S_2$. Now, let $\hat{t} = \max_{s \in \mathcal{L}_1}s.t_c|_{T_1}$. We consider two sub-cases based on $\hat{t}$ in order to determine candidates for $\widetilde{t}$ and $\widetilde{T}$.
	\begin{itemize}
		\item \underline{$\hat{t} > t$}:  In this case, we chose $\widetilde{t} = \hat{t}$, and $\widetilde{T} = T_1$. Once again, no server completes \writetoLtwo$(\widetilde{t}, \widetilde{v})$ before $\widetilde{T}$.
		\item \underline{$\hat{t} = t$}: In this case, let $T^*$ denote the earliest point in the execution when $t'$ is committed by any server in $\mathcal{L}_1$.  We define $\widetilde{t} = \max_{s \in \mathcal{L}_1}s.t_c|_{T^*}$. Also, we define  $\widetilde{T}$ as the earliest point in the execution when $\widetilde{t}$ committed by any server in $\mathcal{L}_1$. Since  $\widetilde{t} \geq t' > t = \hat{t}$, from Lemma \ref{lem:lcasb_prop1} it follows that  $\widetilde{T} > T_1 > T$.  Also, trivially, note that no server in $\mathcal{L}_1$ completes \writetoLtwo$(\widetilde{t}, \widetilde{v})$ before $\widetilde{T}$.
	\end{itemize}
	In either of the two cases above, we can argue like in the proof of Lemma \ref{lem:lcasb_prop3} in order to find the subset $S_b$ that satisfies the lemma.
\end{proof}

\section*{Appendix II: Proofs of Liveness and Atomicity}

\subsection{Proof of Theorem \ref{thm:lcasb_liveness} [{\bf Liveness}]}

\begin{nono-theorem}  
	Consider any well-formed execution of the \calc~algorithm, where at most $f_1 < n_1/2$ and $f_2 < n_2/3$ servers crash fail in layers $\mathcal{L}_1$ and $\mathcal{L}_2$, respectively. Then every operation associated with a non-faulty client completes. 
\end{nono-theorem}
\begin{proof}
	The liveness of write operations is straightforward; we only argue liveness of reads. Consider a read operation $\pi_1$ associated with a non-faulty reader $r$. Let $t_{req}$ denote the tag that was sent by $r$ during the \GetData \ phase. It is easy to see that all servers in $S_a$ definitely respond to reader $r$. Note that each of the responses can either be a valid tag-value pair, or a regenerated tag-coded-element-pair or $(\perp, \perp)$. Let $T_j$ denote the point of execution when server $s_{a, j} \in S_a, 1 \leq j \leq f_1 + k$ acts on the get-data request from $r$. Without loss of generality, let us assume that $T_{i} < T_{i+1}, 1 \leq i \leq f_1 + k -1$.
	
	Consider the sequence ${\bf s}$ of $k$ servers $(s_{a, f_1 + 1}, s_{a, f_1 + 2}, \ldots, s_{a, f_1 + k})$, and suppose that no server in this sequence responds to the reader $r$ as part of the execution of the corresponding \GetDataResp~action; instead all $k$ of them initiate the internal \readfromLtwo$(r)$ operation. Now, let $s_{a, f_1 + i}, 1 \leq i \leq k$ be the last server in this sequence that results in one of the following two events:
	\begin{enumerate}
		\item $s_{a, f_1 + i}$ returns $(\perp, \perp)$
		\item $s_{a, f_1 + i}$ returns a regenerated tag-coded-element pair $(t, c)$, where the tag $t$ is different from the tags returned by the servers $s_{a, f_1 + i + 1}, \ldots, s_{a, f_1 + k}$. Since $s_{a, f_1 + i}$ is taken as the last such server in the sequence, in this case, it is clear that all the servers  $s_{a, f_1 + i + 1}, \ldots, s_{a, f_1 + k}$ indeed regenerate and return a common tag $t'$ such that $t' \geq t_{req}$.
	\end{enumerate}
	First of all note if there is no such server in the sequence ${\bf s}$ that satisfies either of the two conditions above, then it is clear that all $k$ servers respond with with the common regenerated tag $t' \geq t_{req}$, and this ensures liveness. We analyze each of the two events above separately:
	\begin{enumerate}
		\item \emph{$s_{a, f_1 + i}$ returns $(\perp, \perp)$}:  This can happen in two ways:
		\begin{enumerate}
			\item \emph{Server $s_{a, f_1 + i}$ is unable to regenerate any $(t, c)$ pair based on the responses}: In this case, we know from Lemma that \ref{lem:lcasb_prop3} that there exists a point $\widetilde{T} > T_{f_1 + i}$
			in the execution such that the following statement is true: there exists a subset $S_b$ of  $S_a$  such that $|S_b| = k$, and $\forall s' \in S_b$  $(\widetilde{t}, \widetilde{v}) \in s'.L|_{\widetilde{T}}$, where $\widetilde{t} = \max_{s \in \mathcal{L}_1}s.t_c|_{\widetilde{T}}$. Next, note that $|S_b \cap \{s_{a, 1}, s_{a, 2}, \ldots, s_{a, f_1 + i}\}|   \geq k + (f_1 + i) - |S_a| = i \geq 1, 1 \leq i \leq k$. Pick any server $s_{a,j} \in S_b \cap \{s_{a, 1}, s_{a, 2}, \ldots, s_{a, f_1 + i}\}$.  We now consider the two cases $s_{a,j}.t_c|_{\widetilde{T}} = \widetilde{t}$ and $s_{a,j}.t_c|_{\widetilde{T}} < \widetilde{t}$
			and show that in either case, the server $s_{a,j}$ indeed sends a valid tag-value pair $(t, v)$ to the reader such that $t \geq t_{req}$.
			\begin{itemize}
				\item \emph{$s_{a,j}.t_c|_{\widetilde{T}} = \widetilde{t}$}: Let $T$ denote the earliest point in the execution when $s_{a,j}$ commits the tag $\widetilde{t}$. Now, we either have $T < T_j$ or $T > T_j$, where we recall $T_j$ to be the point in the execution when $s_{j}$ acts on the get-data request from $r$. If $T < T_j$, it is clear that $s_{a, j}$ responds to $r$ with $(\widetilde{t}, \widetilde{v})$ as part of the \GetDataResp~action. If $T > T_j$, $s_{j}$ responds (at $T$) to $r$ with $(\widetilde{t}, \widetilde{v})$ as part of the  action where $\widetilde{t}$ gets committed\footnote{It may be noted that $s_{a,j}$ commits  $\widetilde{t}$ via the \broadcastResp~action or via the ``if-clause" of the \PutTagResp~action.  It is not possible that $s_{j}$ commits  $\widetilde{t}$ via the ``else-clause" of \PutTagResp~action, since this would mean that $s_{a,j}$ commits a tag higher than $\widetilde{t}$, earlier than $\widetilde{T}$} (if $s_{j}$ did not yet respond with a valid tag-value pair to $r$.)
				\item \emph{$s_{a,j}.t_c|_{\widetilde{T}} < \widetilde{t}$}: In this case, we claim that there exists a point $T' > \widetilde{T} > T_j$ in the execution such that $t' = s_{a,j}.t_c|_{T'} \geq \widetilde{t}$. The existence of $T'$ follows because of the property of the \emph{broadcast} primitive used in the algorithm, which ensures that the $f_1 + k$ broadcast messages that resulted in some server committing $\widetilde{t}$ will eventually also be received by $s_{a,j}$, and thus $s_{a,j}$'s commitCounter[$\widetilde{t}$] eventually increases to $f_1 + k$ (note that $(\widetilde{t}, \widetilde{v}$) is part of the list of $s_{a,j}$ as discussed above).
				
				Without loss of generality, let $T'$ denote the earliest point in the execution where the above claim holds. Precisely one of the following three actions is associated with $T'$ (the three possibilities arise because of the consideration that the $T'$ is the earliest point for which the claim holds, which results in additional possibilities for $T'$). The server $s_{a,j}$ responds to reader $r$ with $(t,v), t \geq \widetilde{t}$ at $T'$ in any of these three cases, if it did not yet respond to $r$ with a valid tag-value pair.
				\begin{itemize} 
					\item $s_{a,j}$ acts on a  $broadcast(${\reqCommitTag}, $ t')$ message, and updates $t_c$ to $t'$.  In this case, server $s_{a,j}$ responds  with $(t', v')$
					\item $s_{a,j}$ acts on a  $put\act{-}tag((${\putTagLabel}, $t')$ message such that $(t', v') \in s_{a,j}.L|_{T'}$ (thus enters the ``if" clause), and updates $t_c$ to $t'$ .  In this case, server $s_{a,j}$ responds  with $(t', v')$
					\item $s_{a,j}$ acts on a $put\act{-}tag(${\putTagLabel}, $t')$ message such that $(t', v') \notin s_{a,j}.L|_{T'}$ (thus enters the ``else" clause), and updates $t_c$ to $t'$. Observe that in this case $(\widetilde{t}, \widetilde{v})$ never gets removed from list of $s_{a,j}$ before $T'$, and thus server $s_{a,j}$ responds to reader $r$ with $(t,v), t \geq \widetilde{t}$ at $T'$.
				\end{itemize}
			\end{itemize}
			
			\item \emph{Server $s_{a, f_1 + 1}$ regenerates a pair $(t, c)$ based on the responses, however $t < t_{req}$}: This case can be handled like the last one; this time we rely on Lemma \ref{lem:lcasb_prop4} instead of Lemma \ref{lem:lcasb_prop3}.
		\end{enumerate}
		\item \emph{$s_{a, f_1 + i}$ returns a regenerated tag-coded-element pair $(t, c)$, such that $t \neq t'$, where $t'$ is the tag regenerated (and returned) by the servers  $s_{a, f_1 + i + 1}, \ldots, s_{a, f_1 + k}$}. We consider both the cases $t < t'$ and $t > t'$; in either case, we can use  Lemma \ref{lem:lcasb_prop6}, and arguments like in Case $1.(a)$ above to prove liveness of the read operation. 
		
		This completes the proof of liveness.
	\end{enumerate}
\end{proof}

\subsection{Proof of Theorem \ref{thm:lcasb_atomicity} [{\bf Atomicity}]}

Our proof of atomicity is based on Lemma $13.16$ of \cite{Lynch1996} (which gives a sufficient condition for proving atomicity), which is paraphrased below:
\begin{lemma} \label{lem:atom}
	Consider any well-formed execution $\beta$ of the algorithm, such that all the invoked read and 
	write operations  complete. Now, suppose that all  the invoked read and write operations in $\beta$ can be partially ordered by an ordering $\prec$, so that the following properties are satisfied:
	\begin{itemize}
		\item [\em P1.] The partial order ($\prec)$ is consistent with the 
		external order of invocation and responses, i.e., there are no 
		operations $\pi_1$ and $\pi_2$, 
		such that $\pi_1$ completes before $\pi_2$ starts, 
		yet $\pi_2 \prec \pi_1$.
		\item[\em P2.] All operations are totally 
		ordered with respect to the write operations, 
		i.e., if $\pi_1$ is a write operation and $\pi_2$ is any other operation then 
		either $\pi_1 \prec \pi_2$ or $\pi_2 \prec \pi_1$.
		\item[\em P3.] Every read operation returns
		the value of the last write preceding it (with respect to $\prec$), and if no preceding 
		write is ordered before it, then the read returns the initial value
		of the object.
	\end{itemize}
	Then, the execution $\beta$ is atomic.	
\end{lemma}	

Let $\Pi$ denote the set of all successful client operations in $\beta$. Towards defining the partial ordering on $\Pi$, we first recall the definition of tags of read and write operations, associated with non-faulty clients (see Section \ref{sec:propeties}).For a write operation $\pi$,  recall that we defined the $(tag(\pi), value(\pi)))$ pair as the message $(t_w, v)$ which the writer sends in the {\PutData} phase. If $\pi$ is a read, we defined the $(tag(\pi), value(\pi)))$ pair as $(t_{r}, v)$ where $v$ is the value that gets returned, and $t_{r}$ is the  associated tag. The partial order ($\prec$) in $\Pi$ is defined as follows: For any $\pi, \phi \in \Pi$, we say $\pi \prec \phi$  if one of the following holds: $(i)$  $tag(\pi)  < tag(\phi)$, or $(ii)$ $tag(\pi) = tag(\phi)$, and  $\pi$ and $\phi$ are write and read operations, respectively. 

We are now ready to prove Theorem \ref{thm:lcasb_atomicity}.
\begin{nono-theorem}
	Every well-formed execution of the $LDS$ algorithm is atomic.
\end{nono-theorem}
\begin{proof}
	We  prove the atomicity by proving properties $P1$, $P2$ and $P3$ appearing in Lemma \ref{lem:atom} for any execution of the algorithm.
	
	\emph{Property $P1$}: Consider two operations $\pi$ and $\phi$ such that $\pi$ completes before $\phi$ is invoked. We need to show that it cannot be  the case that $\phi \prec \pi$. We will show this for the case when both $\phi$ and $\pi$ are writes. The remaining three cases where at least one of the two operations is a read can be similarly analyzed. All four cases essentially use Lemma \ref{lem:maxL_geq_tag}, which we demonstrate  for the case when both $\phi$ and $\pi$ are writes. Suppose that the writes $\phi$ and $\pi$ are initiated by writers $w_{\phi}$ and $w_{\pi}$, respectively. Let $T$ denote the point of execution when $w_{\phi}$ initiates the \GetTag~phase. Since $\pi$ completes before $T$, we know from Lemma \ref{lem:maxL_geq_tag} that for any set $S'$ of $f_1 + k$ servers in $\mathcal{L}_1$ that are non-faulty at $T$, there exists $s \in  S'$ such that $s.t_c|_T \geq tag(\pi)$ and $\max \{ t: (t, *) \in s.L|_T \}  \geq tag(\pi)$. In this case, if $S'$ denotes the set of $f_1 + k$ servers whose responses were used by  $w_{\phi}$ in the \GetTag~phase (to compute $tag(\phi)$), it is clear that at least one of these responses is at least as high as $tag(\pi)$, and this ensures that $tag(\phi) > tag(\pi)$, and thus $\pi \prec  \phi$.

	\emph{Property $P2$}:  This follows from the construction of tags, and the definition of the partial order ($\prec$).
	
	\emph{Property $P3$}:  This follows from the definition of partial order ($\prec$), and by noting that value returned by a read operation $\pi$ is simply the value associated with $tag(\pi)$.
\end{proof}

\section*{Appendix III: Proofs of Performance Metrics}

\subsection{Proof of Lemma \ref{lem:gctime1} [{\bf Temporary Nature of $\mathcal{L}_1$ Storage}]}

\begin{nono-lemma}
	Consider a successful write operation $\pi \in \beta$. Then, there exists a point of execution $T_e(\pi)$ in $\beta$ such that for all $T' \geq T_{e}(\pi)$ in $\beta$, we have  $s.t_c|_{T'} \geq tag(\pi)$ and $(t, v) \not\in s.L|_{T'}$,  $\forall s \in \mathcal{L}_1, t \leq tag(\pi)$. 
\end{nono-lemma}
\begin{proof}
We first identify point of execution $T_s$ for every non-faulty server $s \in \mathcal{L}_1$ which satisfies the lemma for server $s$. In this case, we can define $T_{e}(\pi) = \max_{s \in \mathcal{L}_1}(T_s)$. Towards this, from our assumption that point-to-point channels are reliable, any (non-faulty) server $s$ in $\mathcal{L}_1$ eventually receives the \PutData~request from the writer, containing the associated (tag, value) pair, say $(t' = tag(\pi), v')$. At this point of execution, say $T$, if we suppose that $s.t_c|_T > t'$, we define $T_s = T$. It follows from Lemma \ref{lem:tctag} that $(t, v) \notin s.L|_{T_e(\pi)}$,  $\forall t \leq tag(\pi)$. The fact that the statement also holds for all $T' > T_s$ for server $s$, as long as it remains non-faulty, follows by combining Lemma \ref{lem:lcasb_prop1} with Lemma \ref{lem:tctag}.
	
Next consider the case where $s.t_c|_T = t'$. From the algorithm, we see that this happens precisely if the server $s$ updated its committed tag to $t'$ at a point $T_1 < T$, via an execution of \PutTagResp(PUT-TAG, $t'$)~action, i.e., some reader returned $(t', v')$ before the point $T$. In this case, it is clear that $(t', v') \notin s.L|_{T}$. In this case also, we define $T_s = T$, and we see that lemma statement holds for server $s$ (as long as it remains non-faulty) by  using Lemma \ref{lem:lcasb_prop1} and Lemma \ref{lem:tctag}. 
	
Finally, consider the case where $s.t_c|_T < t'$. In this case, the server $s$ adds $(t', v')$ to its list $L$ at $T$. Since the write is successful, it received acknowledgments from $f_1 + k$  servers in $\mathcal{L}_1$. From an inspection of the algorithm, we see that all these servers definitely execute the respective \emph{broadcast} primitives associated with this write, before sending the acknowledgment to the writer. Since channels are reliable, server $s$ eventually receives these broadcast messages, and there exits a point in the execution such that commitCounter[t'] $\geq f_1 + k$ for the server $s$. In this case, server $s$ (if it still has a lower $t_c$) initiates \writetoLtwo$(t', v')$, which definitely succeeds. After the termination of \writetoLtwo$(t', v')$, $s$ replaces $(t', v')$ with $(t', \perp)$, if it has not already done so. In this third case, we define $T_s$ as the point of termination of the internal \writetoLtwo~operation. It is straightforward to see that lemma statement holds for server $s$ in this case well.
\end{proof}

\subsection{Proof of Lemma \ref{lem:writecost} [{\bf Write, Read Cost}]}

\begin{nono-lemma}
	The communication cost associated with any write operation in $\beta$ is given by 
	$n_1 + n_1n_2\frac{2d}{k(2d - k + 1)} = \Theta(n_1)$. The communication cost associated with any successful read operation $\rho$ in $\beta$ is given by $n_1(1 + \frac{n_2}{d})\frac{2d}{k(2d-k+1)} + n_1I(\delta_{\rho} > 0) = \Theta(1) + n_1I(\delta_{\rho} > 0)$. Here, $I(\delta_{\rho} > 0)$ is $1$ if $\delta_{\rho} > 0$, and $0$ if $\delta_{\rho} = 0$. 
\end{nono-lemma}
\begin{proof}
	For writing value $v$, the write-cost due to messages exchanged between writer and servers in $\mathcal{L}_1$ is given by $|v|n_1$. The write-cost due to the internal \writetoLtwo~operations is given by $n_1n_2\alpha$. Recall that we are using MBR codes, and the file-size (which is size of $v$ here) of MBR codes is give by $|v| = \sum_{i=0}^{k-1}(d-i)\beta = k\beta(2d - k + 1)/2$. Also recall that MBR codes are characterized  by the relation $\alpha = d\beta$. Using the last two statements, the overall write-cost is given by $|v|n_1(1 + n_2\frac{2d}{k(2d - k + 1)})$. The result about write-cost now follows, if we normalize the cost by the size of $v$.
	
	The read cost due to \readfromLtwo~actions is given by $n_1n_2\beta = n_1n_2\frac{2|v|}{k(2d-k+1)}$. If $\delta_{\rho} = 0$, the read cost due to message exchanges from servers in $\mathcal{L}_1$ to the reader is given by $n_1\alpha = n_1d\beta = n_1\frac{2d|v|}{k(2d-k+1)}$. However, if $\delta_{\rho} > 0$, the worst-case read cost due to message exchanges from servers in $\mathcal{L}_1$ to the reader is given by $n_1\alpha + n_1|v|$. Thus the overall read-cost is given by $n_1n_2\beta + n_1\alpha + n_1|v|I(\delta_{\rho}>0)$. The result once again follows, if we normalize the total cost by the size of $v$.
\end{proof}

\subsection{Proof of Lemma \ref{lem:storcost} [{\bf Single Object Permanent Storage Cost}]}

\begin{nono-lemma}
	The (worst case) storage cost in $\mathcal{L}_2$ at any point in the execution of the \calc~algorithm  is given by $\frac{2dn_2}{k(2d-k+1)} = \Theta(1)$.
\end{nono-lemma}
\begin{proof}
	The cost associated with storing object $v$ in $\mathcal{L}_2$ is given by $n_2\alpha = n_2\frac{2d|v|}{k(2d-k+1)}$. The result follows by normalizing with the size of $v$.
\end{proof}

\subsection{Proof of Lemma \ref{lem:latency} [{\bf Write, Read Latency}]}
\begin{nono-lemma}
	A successful write operation in $\beta$ completes within a duration of $4\tau_1 + 2\tau_0$. The associated extended write operation completes within a duration of $\max(3\tau_1 + 2\tau_0 + 2\tau_2, 4\tau_1 + 2\tau_0)$. A successful read operation in $\beta$ completes within a duration of $\max(6\tau_1 + 2\tau_2, 6\tau_1 + 2\tau_0 + \tau_2)$.
\end{nono-lemma}
\begin{proof}
	A write operation has two phases, and in each phase there is one round of communication between the writer and servers in $\mathcal{L}_1$. However, during the second phase, each server in $\mathcal{L}_1$ before sending ACK to the writer needs to internally receive broadcast messages from $f_1 + k$ servers. The broadcast primitive was briefly discussed in Section \ref{sec:algo}. It we use the implementation in \cite{SODA2016}, a broadcast message is received within a duration of $2\tau_0$. Thus, a write completes within a duration of $4\tau_1 + 2\tau_0$. 
	Towards determining the duration of extended write operation associated with a successful write operation $\pi$, consider the point of execution $T_{e}(\pi)$ as in Lemma \ref{lem:gctime1}.  It is straightforward to see that $T_{e}(\pi)$ is at most $T_{start}(\pi) + 3\tau_1 + 2\tau_0 + 2\tau_2$. The result for duration of extended write $\pi_e$ follows by recalling that $T_{end}(\pi_e) = \max(T_{end}(\pi), T_{e}(\pi))$.
	
	To calculate completion time of read operation, say $\pi$, invoked by a non-faulty reader $r$, we take a re-look at the proof of liveness of read-operations (see Theorem \ref{thm:lcasb_liveness}). Our proof of liveness of reads was a constructive proof, where we identified points in the execution when the servers in $\mathcal{L}_1$ respond to the reader in a manner that guarantees completion of the read operation. The goal here is to find bounds on these points in the  execution, under the bounded latency model. Below, we repeat a lot of steps from liveness proof for sake of clarity. 
	
	As in the proof of liveness, let $t_{req}$ denote the tag that was sent by $r$ during the \GetData \ phase. It is easy to see that all servers in $S_a$ definitely respond to reader $r$. Note that each of the responses can either be a valid tag-value pair, or a regenerated tag-coded-element-pair or $(\perp, \perp)$. Let $T_j$ denote the point of execution when server $s_{a, j} \in S_a, 1 \leq j \leq f_1 + k$ acts on the get-data request from $r$. Without loss of generality, let us assume that $T_{i} < T_{i+1}, 1 \leq i \leq f_1 + k -1$. 
	
	Consider the sequence ${\bf s}$ of $k$ servers $(s_{a, f_1 + 1}, s_{a, f_1 + 2}, \ldots, s_{a, f_1 + k})$, and suppose that no server in this sequence responds to the reader $r$ as part of the execution of the corresponding \GetDataResp~action; instead all $k$ of them initiate the internal \readfromLtwo$(r)$ operation. Now, let $s_{a, f_1 + i}, 1 \leq i \leq k$ be the last server in this sequence that results in one of the following two events:
	\begin{enumerate}
		\item $s_{a, f_1 + i}$ returns $(\perp, \perp)$
		\item $s_{a, f_1 + i}$ returns a regenerated tag-coded-element pair $(t, c)$, where the tag $t$ is different from the tags returned by the servers $s_{a, f_1 + i + 1}, \ldots, s_{a, f_1 + k}$. Since $s_{a, f_1 + i}$ is taken as the last such server in the sequence, in this case, it is clear that all the servers  $s_{a, f_1 + i + 1}, \ldots, s_{a, f_1 + k}$ indeed regenerate and return a common tag $t'$ such that $t' \geq t_{req}$.
	\end{enumerate}
	First of all note if there is no such server in the sequence ${\bf s}$ that satisfies either of the two conditions above, then it is clear that all $k$ servers respond with with the common regenerated tag $t' \geq t_{req}$, and this ensures liveness. \emph{In the bounded latency model, in this case, the read operation completes within a duration of $6\tau_1 + 2\tau_2$.}
	
	We next proceed to find the read duration, if any one of the above conditions indeed hold good. We analyze both the above events separately:
	\begin{enumerate}
		\item \emph{$s_{a, f_1 + i}$ returns $(\perp, \perp)$}:  This can happen in two ways:
		\begin{enumerate}
			\item \emph{Server $s_{a, f_1 + i}$ is unable to regenerate any $(t, c)$ pair based on the responses}: In this case, we know from Lemma that \ref{lem:lcasb_prop3} that there exists a point $\widetilde{T} > T_{f_1 + i}$ in the execution such that the following statement is true: there exists a subset $S_b$ of  $S_a$  such that $|S_b| = k$, and $\forall s' \in S_b$  $(\widetilde{t}, \widetilde{v}) \in s'.L|_{\widetilde{T}}$, where $\widetilde{t} = \max_{s \in \mathcal{L}_1}s.t_c|_{\widetilde{T}}$. Next, note that $|S_b \cap \{s_{a, 1}, s_{a, 2}, \ldots, s_{a, f_1 + i}\}|   \geq k + (f_1 + i) - |S_a| = i \geq 1, 1 \leq i \leq k$. \emph{Using Remark \ref{rem:Ttilde}, under bounded latency, we get that $\widetilde{T} < T_{start}(\pi) + 3\tau_1 + \tau_2$}.
			
			Now, as in the proof of liveness, we pick any server $s_{a,j} \in S_b \cap \{s_{a, 1}, s_{a, 2}, \ldots, s_{a, f_1 + i}\}$, and consider the two cases $s_{a,j}.t_c|_{\widetilde{T}} = \widetilde{t}$ and $s_{a,j}.t_c|_{\widetilde{T}} < \widetilde{t}$. We know that $s_{a,j}$ eventually sends (say at point of execution $T_r$) a valid (tag, value) pair $(t, v)$ to the reader such that $t \geq t_{req}$. If $s_{a,j}.t_c|_{\widetilde{T}} = \widetilde{t}$, we see from the analysis of liveness proof that $T_r < \widetilde{T}$. And if $s_{a,j}.t_c|_{\widetilde{T}} < \widetilde{t}$, we once again infer from the liveness proof that $T_r < \widetilde{T} + 2\tau_0$. The follows, since in this case, the server $s_{a,j}$ satisfies $commitCounter[\widetilde{t}] \geq f_1 + k$ at a point earlier than $\widetilde{T} + 2\tau_0$, and thus can surely respond to the reader via the \broadcastResp~action, if it has not already done so.
			\item \emph{Server $s_{a, f_1 + 1}$ regenerates a pair $(t, c)$ based on the responses, however $t < t_{req}$}: This case can be handled like the last one; this time we rely on Lemma \ref{lem:lcasb_prop4} instead of Lemma \ref{lem:lcasb_prop3}. Even here we can show that $\widetilde{T} < T_{start}(\pi) + 3\tau_1 + \tau_2$, and argue like above.
		\end{enumerate}
		\item \emph{$s_{a, f_1 + i}$ returns a regenerated tag-coded-element pair $(t, c)$, such that $t \neq t'$, where $t'$ is the tag regenerated (and returned) by the servers  $s_{a, f_1 + i + 1}, \ldots, s_{a, f_1 + k}$}. For this case as well, we can show that $\widetilde{T} < T_{start}(\pi) + 3\tau_1 + \tau_2$, and argue like above.
	\end{enumerate}
	
	Thus the server $s_{a, j}$ sends a valid (tag, value) pair (that is sufficient to ensure liveness), in all cases at $T_r < T_{start}(\pi) + 3\tau_1 + \tau_2 + 2\tau_0$, and the response reaches the reader within a duration of $4\tau_1 + \tau_2 + 2\tau_0$ from the start of the read. Also, if we consider any other server whose response was one among the $f_1 + k$ responses needed by the reader to complete the  \GetData~phase, a (tag, coded-element) (or $(\perp, \perp)$) response from this server reaches the reader within a duration of $4\tau_1 + 2\tau_2$, from the start of the read. Thus, the first two phases of the read complete within a duration of $\max(4\tau_1 + 2\tau_2, 4\tau_1 + \tau_2 + 2\tau_0)$. Finally, the third phase of read completes within a duration of $2\tau_1$, and hence the result.
\end{proof}

\subsection{Proof of Lemma \ref{lem:multiobject} [{\bf Relative Cost of Temporary Storage}]}
\begin{nono-lemma}
	At any point in the execution, the worst case storage cost in $\mathcal{L}_1$ and $\mathcal{L}_2$ are upper bounded by $\left\lceil 5 + 2\mu \right\rceil \theta n_1$ and $\frac{2Nn_2}{k+1}$. Specifically, if $\theta << \frac{Nn_2}{kn_1\mu}$, the overall storage cost is dominated by that of permanent storage in $\mathcal{L}_2$, and is given by $\Theta(N)$. 
\end{nono-lemma}
\begin{proof}
	By assumption, we only consider executions with successful writes. Also $\tau_1 = \tau_0$ by assumption, and thus an extended write operation completes within a duration of $5\tau_1 + 2\tau_2 = (5 + 2\mu)\tau_1$. Recall that the definition of extended write operation was motivated by Lemma \ref{lem:gctime1}, and we know that at any point $T$ in the execution after the completion of the extended write, the corresponding  (tag, value) pair is not presented in the temporary storage of any of the servers in $\mathcal{L}_1$. In this case, if $\theta$ denotes the maximum number of concurrent write operations experienced by the system within any duration of $\tau_1$ time-units, it follows that the normalized temporary  storage-cost in $\mathcal{L}_1$ at any point in the execution is at most $\lceil(5 + 2\mu)\rceil\theta n_1$. The storage cost in $\mathcal{L}_2$ at any point is the execution is exactly $Nn_2\alpha$. Since we assume that $f_1 = f_2$ and $n_1 = n_2$, it follows that $d = k$ for the MBR code. In this case, it can be seen that $\alpha = 2|v|/(k+1)$ and thus the normalized storage cost in  $\mathcal{L}_2$ is given by $2Nn_2/(k+1)$. It is clear that if $\theta <<   \frac{Nn_2}{kn_1\mu}$, the overall storage cost is dominated by that of permanent storage in $\mathcal{L}_2$, and is given by $\Theta(N)$.
\end{proof}

\end{document}